 \documentclass[preprint,1p,times,numbers,sort&compress]{elsarticle}

\usepackage{amssymb}
 \usepackage{amsthm}
 \usepackage{amsmath}
 \usepackage{natbib}
 \usepackage{algorithmic}
 \usepackage{algorithm}
 \usepackage{makecell}
 \usepackage{multirow}
 \usepackage{graphicx}
 \usepackage{float}
 \usepackage{subfig}
 \usepackage{booktabs}
 \usepackage{diagbox} 
 \usepackage{hyperref}
 \usepackage{url}
\hypersetup{hidelinks,
	colorlinks=false,
	allcolors=black,
	pdfstartview=Fit,
	breaklinks=true
}

\usepackage{caption}
\usepackage[font=small,labelfont=bf,labelsep=none]{caption}
\usepackage{color,xcolor}

\begin{document}

\begin{frontmatter}



\title{Signed Graph Learning with Hidden Nodes}


\author[1]{Rong Ye}
\ead{1212037@mail.dhu.edu.cn}

\author[1]{Xue-Qin Jiang \corref{cor1}}
\ead{xqjiang@dhu.edu.cn}

\author[2,3]{Hui Feng \corref{cor1}}
\ead{hfeng@fudan.edu.cn}
\author[4]{Jian Wang}
\ead{jian\_wang@fudan.edu.cn}
\author[1]{Runhe Qiu}
\ead{qiurh@dhu.edu.cn}
\cortext[cor1]{Corresponding author}

\address[1]{College of Information Science and Technology, Donghua University, Shanghai 201620, China}
\address[2]{School of Information Science and Technology, Fudan University, Shanghai 200433, China}
\address[3]{State Key Laboratory of Integrated Chips and Systems, Fudan University, Shanghai 200433, China}
\address[4]{School of Data Science, Fudan University, Shanghai 200433, China}
%
%

\begin{abstract}
Signed graphs, which are characterized by both positive and negative edge weights, have recently attracted significant attention in the field of graph signal processing (GSP). Existing works on signed graph learning typically assume that all graph nodes are available.
However, in some specific applications, only a subset of nodes can be observed while the remaining nodes stay hidden. To address this challenge, we propose a novel method for identifying signed graph that accounts for hidden nodes, termed \textit{signed graph learning with hidden nodes under column-sparsity regularization} (SGL-HNCS). 
Our method is based on the assumption that graph signals are smooth over signed graphs, i.e., signal values of two nodes connected by positive (negative) edges are similar (dissimilar). Rooted in this prior assumption, the topology inference of a signed graph is formulated as a constrained optimization problem with column-sparsity regularization, where the goal is to reconstruct the signed graph Laplacian matrix without disregarding the influence of hidden nodes. We solve the constrained optimization problem using a tailored block coordinate descent (BCD) approach. Experimental results using synthetic data and real-world data demonstrate the efficiency of the proposed SGL-HNCS method.

\end{abstract}


\end{frontmatter}

\section{Introduction}
{I}{n} numerous scientific disciplines, graphs serve as crucial mathematical tools for modeling complex systems encountered in real-world applications, including social networks \cite{cam1}, sensor networks \cite{sand2}, brain networks \cite{lie3}, and financial networks \cite{An4}.
A graph structured model consists of sets of nodes and edges, where nodes act as entities of the graph and edges model pairwise relationships between the nodes. 
When attributes are assigned to such nodes and modeled as signals residing on the graph, this graph-based data representation gives rise to the emerging filed of graph signal processing (GSP) \cite{Shuman5, sand6, Marques7, leus2023, yan2022, liu2024}. Indeed, GSP is used in several applications such as computer vision \cite{gir2021, dinesh}, chemistry\cite{song} and ecology \cite{jin, li2025}.

An associated challenge in GSP is graph learning (GL) \cite{Pavez8, gian9, Dong10, Mate11,song2022}, which aims to infer the structure of underlying graph from nodal observations. So far, a significant amount of literature has been developed for this task, which can be mainly classified into stochastic and deterministic approaches. Stochastic approaches rely on statistical inference, where graph signals are modeled using a Laplacian constrained Gaussian-Markov random filed (GMRF) \cite{egi12, kum13,jav}. In this model, the problem of interest corresponds to estimating the precision matrix. In addition, non-probabilistic graph learning approaches have also been explored, typically imposing specific assumptions for valid graph estimation. A common assumption is the smoothness \cite{Vk14, 15, baghe}, which is widely used for its simplicity. The stationary assumption \cite{seg16} is another GSP-based property that restricts the spectral characteristics of the estimated graph. However, these approaches are limited to learning unsigned graphs, whose edge weights are strictly positive. 

In contrast, signed graphs, i.e., the graphs whose edge weights can be either positive or negative, have attracted a lot of attention in recent years \cite{dit, matz2020}. A typical example is found in social networks \cite{gir} (see Fig. \ref{f1.1}). The relationships  between users can be represented by positive edges for friendships and negative edges for non-friendships.
Similarly, the model and method can also be applied in recommendation systems to determine whether different users like or dislike a product. 
The $\text{scSGL}$ method proposed in \cite{kara18} demonstrates that a signed graph can be decomposed into two unsigned graphs based on the signs of the edges. Subsequently, the signed graph is learned under the assumption that graph signals admit low-frequency representation over positive edges, while admitting high-frequency representation over negative edges.

Note that the aforementioned graph learning works, whether for unsigned or signed graphs, typically assume that observations of all nodes are fully available. However, in practical scenarios, the graph signals are often observed from only a limited subset of the original nodes, while the rest remain hidden. 
Intuitively, ignoring the existence of hidden nodes can significantly degrade the performance of existing graph learning methods. To address this limitation, several studies have considered learning graph structures with hidden nodes, including works on Gaussian graphical model selection \cite{Chan19, Chang20, Yang21}, linear Bayesian model \cite{Ana22}, and nonlinear regression \cite{Mei23}. Furthermore, some effective graph learning methods based on the context of GSP have also accounted for hidden nodes. For instance, two related works \cite{Bu26} and \cite{Bu27} propose leveraging, respectively, smoothness prior \cite{Shuman5, giraldo2022} or stationary prior \cite{Ag24, Per25} to infer graph topology from incomplete data. The authors in \cite{Rey28} address the problem of learning multi-layer graphs with hidden nodes, assuming that the observed graph signals followed a GMRF model. In our previous work \cite{ye29}, we propose a framework for learning time-varying graphs with hidden nodes based on graph smoothness and stationarity assumptions. However, these methods are primarily designed for learning unsigned graphs with hidden nodes. It remains unclear how existing methods can be extended to signed graph structures with hidden nodes. A real-world example of such structures arises in recommendation networks, where users may decline to disclose whether they like or dislike certain products, leading to hidden relationships. 
Consequently, modeling the influence of hidden nodes in the context of signed graph learning becomes crucial. For clarity, Table \ref{tab1} summarizes the proposed method and related graph learning methods based on GSP framework. 
\begin{table*}[t]
	\caption{ List of proposed method and alternative.  }\label{tab1}
	\centering
	\setlength
	\tabcolsep{5pt}
	\vspace{10pt}
	\renewcommand
	\arraystretch{2.2}
	\begin{tabular}{   c | c| c | c }
		\Xhline{1px}
		\textbf{\makecell[c]{Graph Type}} & 
		\textbf{Method}
		
		& \textbf{Approach} &
		\textbf{\makecell[c]{Hidden Nodes}} \\
		\Xhline{1px}
		
		\multirow{4}{*}{\makecell[c]{Static \\ unsigned graph}}&\text{GL\cite{15}}&\text{Signal smoothness}& \text{-} \\ \cline{2-4}
		
		\multirow{4}{*}{}& GS-Rw \cite{Bu26}  
		& \text{Signal stationarity}	
		& \checkmark  \\ \cline{2-4}
		\multirow{4}{*}{ }& GSm-GL \cite{Bu27}  
		& \text{Signal smoothness}	
		& \checkmark \\ \cline{2-4}
		\multirow{4}{*}{ }& GSm-St-GL \cite{Bu27}  
		& \makecell[c]{Signal smoothness \\ and stationarity}	
		& \checkmark 
		\\ \Xhline{1px}
		
		\makecell[c]{Multi-layer \\ unsigned graphs} & \text{PGL \cite{Rey28} }  &  Signal stationarity
		& \checkmark\\ \Xhline{1px}
		
		\makecell[c]{Time-varying \\ unsigned graphs} & TGSm-St-GL \cite{ye29}   &  \makecell[c]{Signal smoothness \\ and stationarity}
		& \checkmark\\ \Xhline{1px}
		
		\multirow{3}{*}{\makecell[c]{Static \\ signed graph}} 
		\multirow{3}{*}{}&\text{scSGL\cite{kara18}}&\text{Signal smoothness}&\text{-}\\ \cline{2-4}
		\multirow{3}{*}{}& SGL-HNCS(proposed) &\makecell[c]{Signal smoothness with\\ column-sparsity regularization}&\checkmark\\ \Xhline{1px}
	\end{tabular}
\end{table*}

\begin{table}[!t]
	\caption{ List of notations and their description.  }\label{tab2}
	\centering
	\setlength
	\tabcolsep{6pt}
	\vspace{6pt}
	\renewcommand
	\arraystretch{1.4}
	\begin{tabular}{  c | c  }
		\hline\hline
		\textbf{Notation}&\textbf{Description}\\
		\hline\hline
		$\Vert \mathbf{A}\Vert_F^2$& {sum of squared values of all elements of matrix $\mathbf{A}$ } \\
		\hline
		$\Vert\mathbf{A}\Vert_*$ & nuclear norm of matrix $\mathbf{A}$\\
		\hline
		$\Vert \mathbf{A}\Vert_{2,1}$ & sum of the $\mathnormal{l}_2$ \text{norms of all  columns of matrix $\mathbf{A}$} \\
		\hline
		$\mathrm{tr}(\mathbf{A})$ & \text{trace operator of matrix $\mathbf{A}$}\\
		\hline
		$\mathrm{diag}(\mathbf{A})$ & {a vector formed by diagonal elements of matrix $\mathbf{A}$}\\ \hline
		$\mathrm{upper}(\mathbf{A})$ & {a vector formed by upper triangular part of matrix $\mathbf{A}$}\\ \hline
		$\mathrm{vec}(\mathbf{A})$ & {a vector formed by all elements of matrix $\mathbf{A}$}\\ \hline
		$\text{Im}(\mathbf{A})$ $|$ $\mathbf{I}$ & \text{ image of matrix $\mathbf{A}$ $|$ identity matrix}\\ \hline
		$ \mathbf{0\ | \ 1}$ & \text{  column vector of zeros $|$ column vector of ones } \\
		\hline
		$\mathcal{B}$ & {a set of $B$ observable nodes  } \\
		\hline
		$\mathcal{H}$& {a set of $H$ hidden nodes}\\
		\hline
		$\mathcal{V}$& {a set of $N=B+H$ full nodes}\\
		\hline
		$\mathbf{X}$ & \text{ a matrix of graph signals }\\ \hline
		$\hat{\mathbf{C}} $ & {a sample covariance matrix of graph signals $\mathbf{X}$}\\ \hline
		$\mathcal{G}^{+} \ | \ \mathcal{G}^{-}$ & {a graph with positive edge weights \ $|$ \ a graph with negative edge weights}\\ \hline
		$\mathbf{L}^{+} \ | \ \mathbf{L}^{-}$ &{the Laplacian matrix of graph $\mathcal{G}^{+}$ \ $|$ \ the Laplacian matrix of graph $\mathcal{G}^{-}$}\\ \hline
		$\mathbf{W}\ | \ \mathbf{D}$ & {the adjacency matrix \ $|$ \ the degree matrix}\\ \hline
		$\mathcal{L}$ & {the set of combination Laplacian matrix}\\ \hline
		$\bar{\mathcal{L}}$ & {the set of non-combination Laplacian matrix}\\ \hline
	\end{tabular}
\end{table}

To this end, we propose a \textit{signed graph learning with hidden nodes under column-sparsity regularization} (SGL-HNCS) method for learning a signed graph while taking into account the presence of hidden nodes. Based on the smoothness assumption, we formalize the relationship between the observable nodes
and the unknown signed graph under the influence of the hidden
nodes.
The primary contributions can be summarized as follows:
\begin{itemize}
	\item{We formulate the problem of learning signed graph with hidden nodes by leveraging the concepts of node similarity and dissimilarity from \cite{kara18}. Under this setting, the problem is cast as a constrained optimization problem with column-sparsity regularization.}
	\item{Our learning algorithm is based on the well-known block coordinate descent (BCD) approach. The algorithm decomposes the problem into three subproblems, allowing for joint estimation of our model parameters. Specifically, we introduce the alternating direction method of multipliers (ADMM) algorithm to efficiently address the first non-convex subproblem. Furthermore, we provide rigorous proofs for the convergence of the proposed algorithm. }
	\item{Experimental results demonstrate that the proposed method exhibits higher accuracy in signed graph learning compared to those of other state-of-the-art methods.}	
\end{itemize}

The remainder of this paper is organized as follows. Section \ref{sec2} provides a comprehensive review of fundamental concepts related to GSP and an overview of the associated graph learning methods. Section \ref{sec3} formally introduces the problem of learning signed graph with hidden nodes and formulates the problem. Section \ref{sec4} discusses the algorithm development. Section \ref{sec5} analyzes the associated convergence guarantees of the proposed algorithm. Experimental results are presented in Section \ref{sec6}. Finally, we conclude the paper in Section \ref{sec7}.

\textbf{Notations}: Throughout the paper, scalars are denoted by normal lower-case letters (e.g., $a$), vectors by boldface lower-case letters (e.g., $\mathbf{a}$) and matrices by boldface upper-case letters (e.g., $\mathbf{A}$). Calligraphic font capital letters are used to denote sets (e.g., $\mathcal{S}$). The notation $\mathbb{R}^{N\times N}$ represents the set of matrices of size $N\times N$. The rest of notations are summarized in Table \ref{tab2}.

\section{Preliminaries}\label{sec2}
This section serves as a foundational introduction. We first present fundamental definitions related to graph models and GSP. Furthermore, we introduce the frameworks for learning unsigned and signed graphs based on smooth graph signal models, respectively. 
\subsection{Graph and Graph Signals}
A weighted undirected graph with $N$ nodes can be represented by $\mathcal{G}=(\mathcal{V},\mathcal{E},\mathbf{W})$, with $\mathcal{V}=\{1,\cdots,N\}$ being the set of graph nodes, and $\mathcal{E}\subseteq\{\{i,j\}|i,j\in\mathcal{V}\}$ being the set of graph edges, $i$ and $j$ representing two different nodes. The matrix $\mathbf{W}$ is a weighted adjacency matrix whose elements represent the weight of the edges. For an undirected graph, $\mathbf{W}$ is symmetric with zero diagonal entries, and the $(i,j)$-th entry $W_{ij}$ is assigned a non-negative value if $(i,j)\in\mathcal{E}$. Define the Laplacian matrix as $\mathbf{L: =D-W}$, where $\mathbf{D}$ is the diagonal degree matrix with entries $D_{ii}=\sum_{j=1}^N W_{ij}$ being the degree of node $i$ and $D_{ij}=0$ for $i\ne j$. The Laplacian matrix $\mathbf{L}$ can be decomposed into $\mathbf{L}=\mathbf{U\Lambda U}^\top$ due to its symmetry, where superscript $^\top$ is a transpose operator, $\mathbf{U}=[\mathbf{u}_{1},\cdots,\mathbf{u}_{N}]\in \mathbb{C}^{N\times N}$ is a matrix consisting of the eigenvectors of $\mathbf{L}$, and $\mathbf{\Lambda}\in\mathbb{C}^{N\times N}$ is a diagonal matrix containing the corresponding eigenvalues arranged in increasing order. Let a vector $\mathbf{x}=[x_{1},\cdots, x_{N}]^\top\in\mathbb{R}^{N}$ be a graph signal defined on the $\mathcal{G}$, where $x_{i}$ denotes the signal value at node $i$.

\begin{figure*}[!t]
	\centering
	\subfloat[Signed graph $\mathcal{G}$]{\includegraphics[width=0.32\columnwidth]{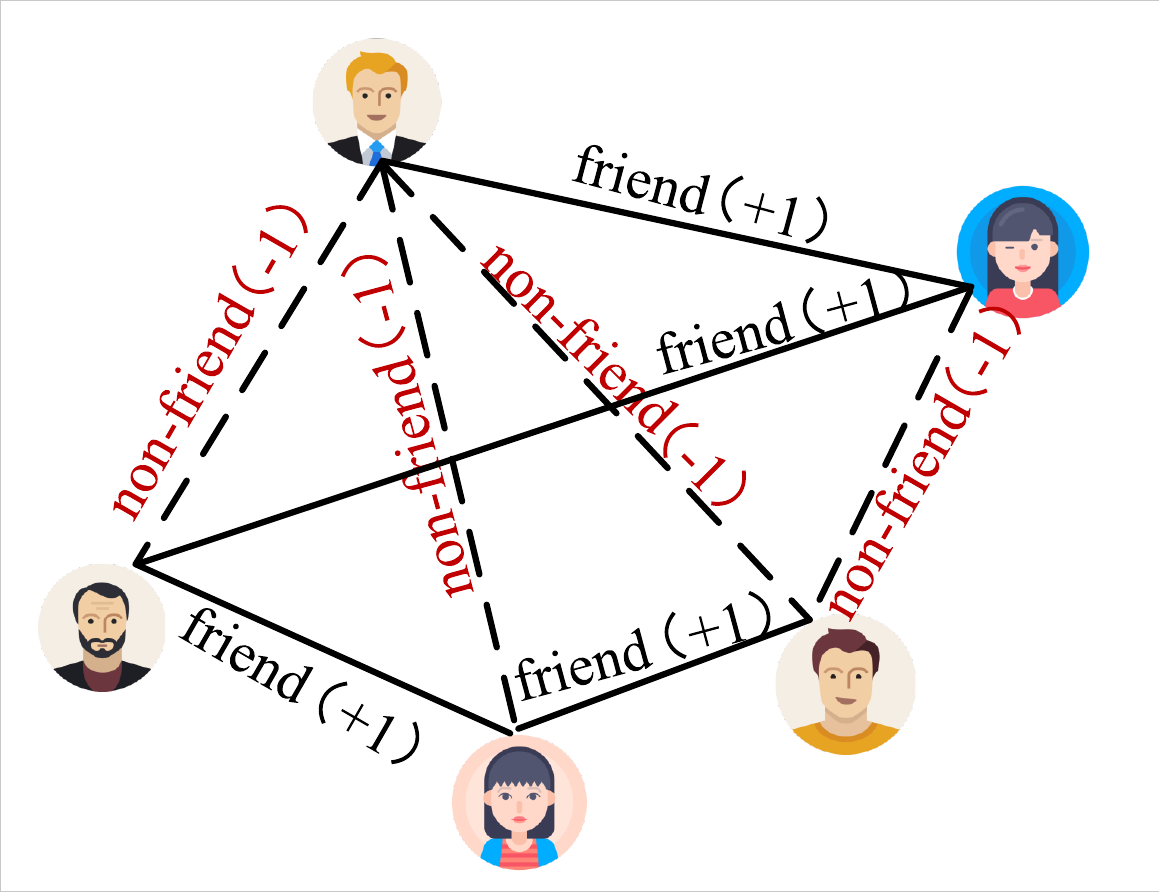}
		\label{f1.1}}\hfil
	\subfloat[Unsigned graph $\mathcal{G}^+$]{\includegraphics[width=0.32\columnwidth]{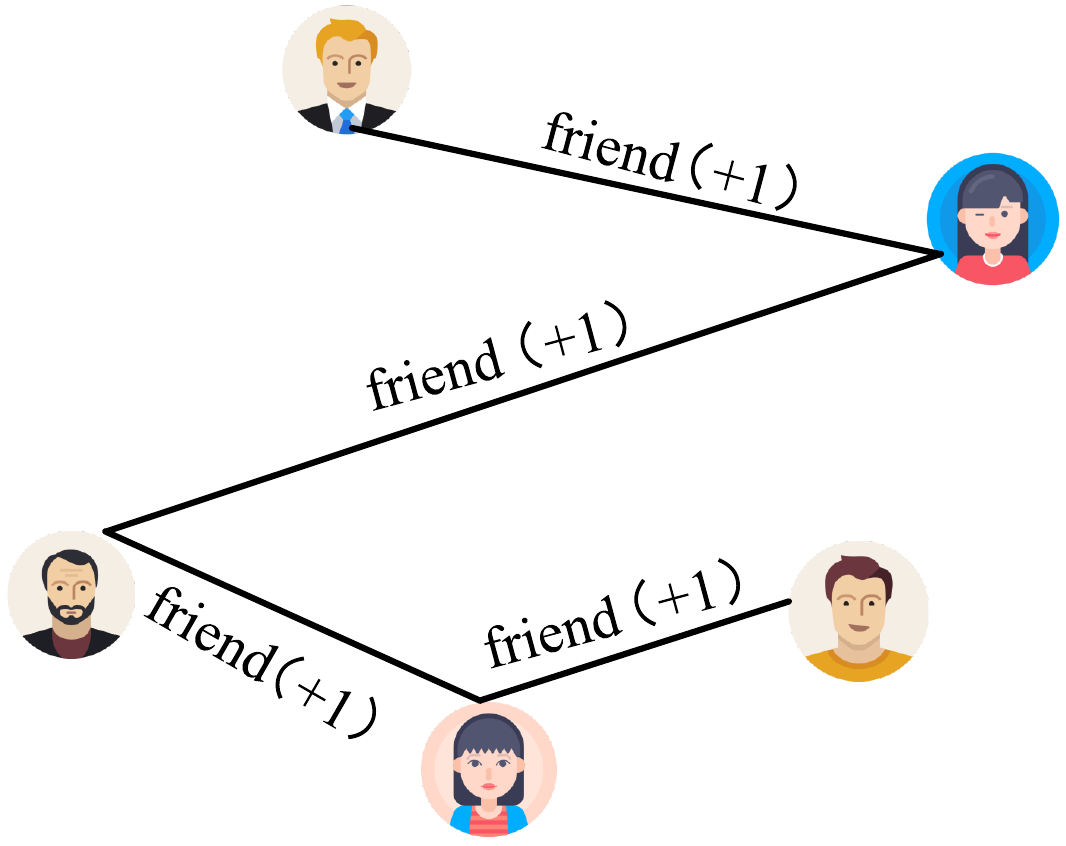}
		\label{f1.2}}\hfil
	\subfloat[Unsigned graph $\mathcal{G}^-$]{\includegraphics[width=0.32\columnwidth]{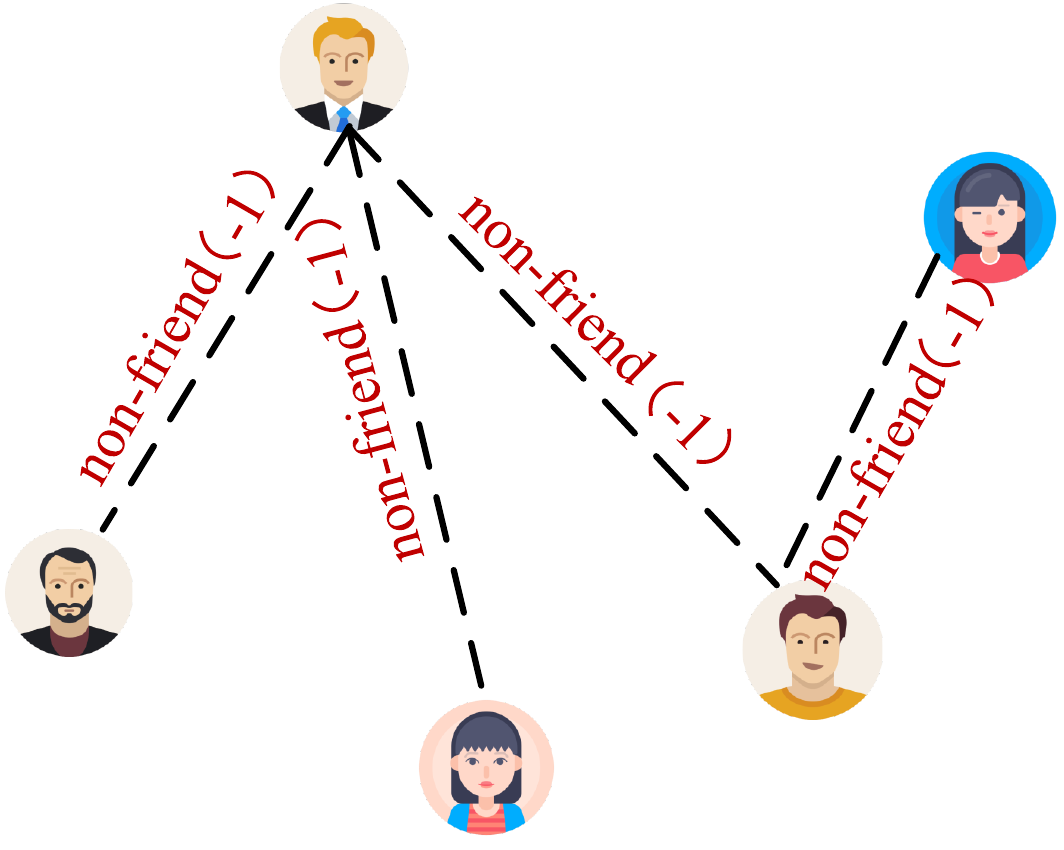}
		\label{f1.3}}\\
	\subfloat[Adjacency matrix of $\mathcal{G}$ ]{\includegraphics[width=0.33\columnwidth]{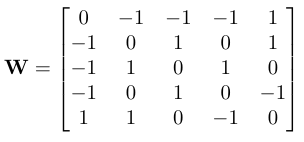}
		\label{f1.4}}\hfil
	\subfloat[Adjacency matrix of $\mathcal{G}^+$ ]{\includegraphics[width=0.3\columnwidth]{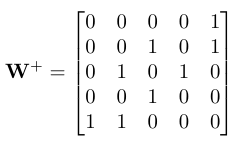}
		\label{f1.5}}\hfil
	\subfloat[Adjacency matrix of $\mathcal{G}^-$]{\includegraphics[width=0.3\columnwidth]{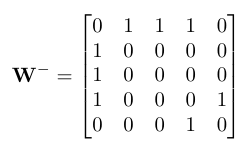}
		\label{f1.6}}

	\caption{An example of decomposing a signed graph into two unsigned graphs $\mathcal{G}^+$ and $\mathcal{G}^-$ in a social network. }
	\label{social_network}
\end{figure*}

If $\mathcal{G}$ is a signed graph, it can be decomposed into two unsigned graphs $\mathcal{G}^+=(\mathcal{V},\mathcal{E}^+,\mathbf{W}^+)$ and $\mathcal{G}^-=(\mathcal{V},\mathcal{E}^-,\mathbf{W}^-)$ \cite{kara18}. Here, $\mathcal{E}^+$ and $\mathcal{E}^-$ denote the edge set of graphs $\mathcal{G}^+$ and $\mathcal{G}^-$, respectively. The entries of $\mathbf{W}^+$ and $\mathbf{W}^-$ characterize the strength of the connections between nodes $i$ and $j$. Specifically, we have ${W}_{ij}^{+}={W}_{ij}$ for ${W}_{ij}>0$ and ${W}_{ij}^{+}=0$ for ${W}_{ij}\leq0$, and ${W}_{ij}^{-}=\left|{W}_{ij}\right|$ for ${W}_{ij}<0$ and ${W}_{ij}^{-}=0$ for ${W}_{ij}\geq0$. An illustration of this decomposition process is shown in Fig \ref{social_network}. Based on the definition of $\mathbf{L}$, two Laplacian matrices $\mathbf{L}^{+}$ and $\mathbf{L}^-$ are constructed from $\mathbf{W}^{+}$ and $\mathbf{W}^-$, respectively.
\subsection{Unsigned Graph Learning from Smooth Signals}
The $K$ graph signals are represented by the vectors $\{\mathbf{x}_i\}_{i=1}^K$, where $\mathbf{x}_i\in\mathbb{R}^N$. These graph signals are collected in a $N\times K$ data matrix $\mathbf{X}=[\mathbf{x}_1,\cdots,\mathbf{x}_{K}]$. An unknown unsigned graph can be learned based on assumptions regarding the relationship between the graph signals and the graph structure. Assuming that the graph signals are smooth with respect to the corresponding graph, the work in \cite{15} proposes a general graph learning framework. The smoothness of $\mathbf{X}$ can be quantified by Laplacian quadratic form. Thus, the method in \cite{15} addresses the following convex optimization problem
\begin{align}
	\label{ex1}
	&\min \limits_{\mathbf{L}\in\mathcal{L}} \quad \mathrm{tr}(\mathbf{X^\top}\mathbf{L}\mathbf{X})+\alpha\left\|\mathbf{L}\right\|_F^2\notag\\
	&\mathrm{s.t.}\quad
	\mathrm{tr}(\mathbf{L})=N,
\end{align}
where $\mathcal{L}:=\{L_{ij}\leq0,i\ne j; \mathbf{L}=\mathbf{L}^\top;\mathbf{L1}=\mathbf{0};\mathbf{L}\succeq{0}\}$ is the set of valid combinatorial Laplacian, and $\alpha$ is a positive regularization parameter. The first term in the objective function quantifies the total variation of $\mathbf{X}$, measuring the smoothness of the graph signals on unsigned graph. The second term acts as a sparsity-promoting regularization on $\mathbf{L}$.  Additionally, the trace constraint guarantees that the trivial solutions are avoided.
\subsection{Signed Graph Learning from Smooth Signals}
Inspired by unsigned graph learning method, an approach to learn signed graph from smooth signals is proposed in \cite{kara18}. The signed graph $\mathcal{G}$ is learned based on the following assumptions: (i) the graph signals exhibit small total variation with respect to $\mathcal{G}^+$; (ii) the graph signals exhibit large total variation with respect to $\mathcal{G}^-$. More specifically, given a data matrix $\mathbf{X}$, the signed graph learning method proposed in \cite{kara18} is formulated as minimization
\begin{align}
	\label{ex2}
	\min\limits_{\mathbf{L}^+,\mathbf{L}^{-}\in\mathcal{L}}\quad
	&\mathrm{tr}(\mathbf{X}^\top\mathbf{L}^{+}\mathbf{X})-\mathrm{tr}(\mathbf{X}^\top\mathbf{L}^{-}\mathbf{X})+\alpha_{1}\left\|\mathbf{L}^{+}\right\|_F^2+\alpha_{2}\left\|\mathbf{L}^{-}\right\|_F^2\notag\\
	\mathrm{s.t.}\quad
	&\mathrm{tr}(\mathbf{L}^{+})=N, \mathrm{tr}(\mathbf{L}^{-})=N,\notag\\ &(\mathbf{L}^+,\mathbf{L}^-)\in\mathcal{U},
\end{align} 
where $\alpha_1 \textgreater 0$ and $\alpha_2 \textgreater 0$ are constant parameters. The last constraint ensures that $\mathbf{L}^+$ and $\mathbf{L}^-$ do not have non-zero entries at the same indices, where set $\mathcal{U}=\{(\mathbf{L}^+,\mathbf{L}^-):L_{ij}^+=0\ \text{if}\ L_{ij}^-\ne0\ \text{and}\ L_{ij}^-=0\ \text{if}\ L_{ij}^+\ne0,\forall {i\ne j}\} $.

\section{Signed Graph Learning with Hidden Nodes }\label{sec3}
In this section, we consider the case where the graph signals are observed only from a subset of nodes during the data collection process. In Subsection \ref{sec3.1}, we analyze the influence of hidden nodes on the signed graph structure inference and provide a detailed description of the resulting problem. In Subsection \ref{sec3.2}, we formulate the signed graph learning problem as a constrained optimization problem. 
\subsection{Influence of Hidden Nodes in Signed Graph}\label{sec3.1}
In Section \ref{sec2}, we introduced the fundamental definitions of signed graph under the assumption that all nodes are observable. However, topology identification becomes complicated when hidden nodes are present. In this subsection, we first discuss some definitions of signed graph while considering hidden nodes, which will be exploited throughout this paper.

Formally, let $\mathbf{X}=[\mathbf{x}_1,\cdots,\mathbf{x}_K]\in\mathbb{R}^{N\times K}$ denote a collection of $K$ graph signals defined on the $N$ nodes of the unknown signed graph $\mathcal{G}$. Given the presence of hidden nodes, only a subset of nodes can be observed.
To accommodate
this setting, we partition the set of nodes $\mathcal{V}$ into two disjoint subsets $\mathcal{B}$ and $\mathcal{H}$. The subset $\mathcal{B}$ is the set of observable nodes and $\mathcal{H}$ is the set of hidden nodes with $\mathcal{H}=\mathcal{V}\setminus\mathcal{B}$, where $\setminus$ represents the relative difference between sets. In particular, we set $\mathcal{B}=\{1,\cdots,B\}$ with cardinality $\left|\mathcal{B}\right|=B$ and $\mathcal{H}=\{B+1,\cdots,N\}$ with cardinality $\left|\mathcal{H}\right|=H=N-B$. By separating the observable and hidden nodes, the complete graph signals are represented as $\mathbf{X}=[\mathbf{X}_B^\top,\mathbf{X}_H^\top]^\top$, where $\mathbf{X}_B\in\mathbb{R}^{B\times K}$ denotes the graph signal values of the observable nodes, corresponding to the first $B$ rows of $\mathbf{X}$, and $\mathbf{X}_H\in\mathbb{R}^{H\times K}$ denotes the graph signal values of the hidden nodes. Furthermore, $\hat{\mathbf{C}}:=\frac{1}{K}\mathbf{X}\mathbf{X}^\top$ stands for the sample covariance matrix of $\mathbf{X}$. To this end, the matrix $\hat{\mathbf{C}}$ and the Laplacian matrix $\mathbf{L}^{s}$ for all $s\in\{+,-\}$ can be expressed by matrix decomposition as 

\begin{equation}\label{ex3}
	\begin{gathered}	
		\hat{\mathbf{C}}=\begin{bmatrix}\hat{\mathbf{C}}_{B} & \hat{\mathbf{C}}_{BH}\\ \hat{\mathbf{C}}_{HB} & \hat{\mathbf{C}}_{H} 
		\end{bmatrix},  
		\mathbf{L}^{s}=\begin{bmatrix} \mathbf{L}_{B}^{s} &        \mathbf{L}_{BH}^{s} \\ \mathbf{L}_{HB}^{s} & \    \mathbf{L}_{H}^{s}
		\end{bmatrix}, \forall s\in\{+,-\}.
	\end{gathered}
\end{equation}
Note that the sample covariance matrix of graph signals residing on the observable nodes is represented by $\hat{\mathbf{C}}_B=\frac{1}{K}\mathbf{X}_B\mathbf{X}_B^\top$. Furthermore, the dependencies among the observable nodes, between the observable and hidden nodes and among the hidden nodes are described by submatrices $\mathbf{L}_B^s, \mathbf{L}_{BH}^s,\mathbf{L}_H^s$, respectively. Throughout this paper, the signed graph $\mathcal{G}$ is decomposed into two unsigned and undirected graphs, $\mathcal{G}^+$ and $\mathcal{G}^-$. Consequently, both matrices $\hat{\mathbf{C}}$ and $\mathbf{L}^s$ are symmetric. Similarly, the submatrices $\hat{\mathbf{C}}_{BH}$ and $\mathbf{L}_{BH}^s$ exhibit the same symmetry, such that $\hat{\mathbf{C}}_{BH}=(\hat{\mathbf{C}}_{HB})^\top$ and $\mathbf{L}_{BH}^s=(\mathbf{L}_{HB}^s)^\top$ for all $s\in\{+,-\}$.

The block structure of the matrices in (\ref{ex3}) motivates the search for an optimal signed graph while considering the presence of hidden nodes. With these foundations in place, we can now formally state the problem of learning  signed graph in the presence of hidden nodes.

\textbf{Problem}: Given a known nodal subset $\mathcal{B}\subset\mathcal{V}$ and a ${B}\times{K}$ observed measurement matrix $\mathbf{X}_{B}$ (where each
row corresponds to an observable node), we aim to estimate the Laplacian matrices $\mathbf{L}_B^+$ and $\mathbf{L}_B^-$ under the assumptions that

\textbf{AS1.} The number of hidden nodes $H$ is significantly smaller than the number of observable nodes $B$ with cardinality $H \ll B$;

\textbf{AS2.} The full graph signals $\mathbf{X}$ have small total variation on the unsigned graph $\mathcal{G}^+$;

\textbf{AS3.} The full graph signals $\mathbf{X}$ have large total variation on the unsigned graph $\mathcal{G}^-$.

The above three assumptions are imposed to make the problem of learning signed graph with hidden nodes more tractable. The first assumption \textbf{AS1} ensures that the availability of information from most nodes. The second assumption \textbf{AS2} and the third assumption \textbf{AS3} establish a relationship between the unknown signed graph structure and the graph signals.

\subsection{Problem Formulation}\label{sec3.2}
We first revisit the measure of smoothness of signals on unsigned graph, i.e., $\frac{1}{K}\mathrm{tr}(\mathbf{X}^\top\mathbf{L}\mathbf{X})$. Given that $\hat{\mathbf{C}}=\frac{1}{K}\mathbf{X}\mathbf{X}^\top$, the total variation of graph signals can be rewritten as $\mathrm{tr}(\hat{\mathbf{C}}\mathbf{L})$. Similarly, the total variation of the graph signal on the signed graph $\mathcal{G}$ can be expressed in the general form $\mathrm{tr}(\hat{\mathbf{C}}\mathbf{L}^s)$ for all $s\in\{+,-\}$, which encompasses the total variation of graph signals on the graph $\mathcal{G}^+$ (i.e., $\mathrm{tr}(\hat{\mathbf{C}}\mathbf{L}^+)$) and on the graph $\mathcal{G}^-$ (i.e., $\mathrm{tr}(\hat{\mathbf{C}}\mathbf{L}^-)$).
It is important to note that due to the existence of hidden nodes, we may not have access to the full sample covariance matrix $\hat{\mathbf{C}}$, but only the observed sample covariance matrix $\hat{\mathbf{C}}_B$. In this sense, the block structure of matrix defined in (\ref{ex3}) is used to rewrite the total variation of graph signals on $\mathcal{G}$, leading to
\begin{align}\label{ex4}
	\mathrm{tr}(\hat{\mathbf{C}}\mathbf{L}^s)
	&=\mathrm{tr}(\hat{\mathbf{C}}_{B}\mathbf{L}_{B}^s)+2\mathrm{tr}\big(\hat{\mathbf{C}}_{BH}(\mathbf{L}_{BH}^s)^\top\big)+\mathrm{tr}(\hat{\mathbf{C}}_H\mathbf{L}_H^s) \notag\\
	&=\mathrm{tr}(\hat{\mathbf{C}}_{B}\mathbf{L}_{B}^s)+2\mathrm{tr}(\mathbf{P}^s)+\mathrm{tr}(\mathbf{R}^s),
\end{align}
where $\mathbf{P}^s :=\hat{\mathbf{C}}_{BH}(\mathbf{L}_{BH}^s)^\top\in\mathbb{R}^{B\times{B}}$, $\mathbf{R}^s :=\hat{\mathbf{C}}_H\mathbf{L}_H^s\in\mathbb{R}^{H\times{H}}, \forall s\in\{+,-\}$. Under the assumption \textbf{AS1}, the matrix $\mathbf{P}^s$ is evidently a low-rank matrix, satisfying $\mathrm{rank}(\mathbf{P}^s)\leq H\ll B$. By defining the matrices $\mathbf{P}^s$ and $\mathbf{R}^s$, the problem that most of the submatrices associated with the hidden nodes are unknown can be solved.

However, the challenge arises from the fact that the matrix $\mathbf{L}_{B}^s$ in (\ref{ex4})
is not a combinatorial Laplacian, but belongs to the set $\bar{\mathcal{L}} :=\{ L_{ij}^s\leq{0}, i\ne j;
\mathbf{L}^s=(\mathbf{L}^s)^\top;\mathbf{L}^{s}\mathbf{1}\ge{\mathbf{0}};\mathbf{L}^s\succeq{0}, \forall{s\in\{+,-\}}\}$. The key difference between the set $\bar{\mathcal{L}}$ and the set of valid combinatorial Laplacian $\mathcal{L}$, is that the condition $\mathbf{L}\mathbf{1}=\mathbf{0}$ is replaced by $\mathbf{L}^{s}\mathbf{1}\ge{\mathbf{0}}$, while others remain unchanged. This adjustment accounts for the negative entries in $\mathbf{L}_{BH}^s$, which arise due to the connections between the nodes in set $\mathcal{B}$ and those in set $\mathcal{H}$. In other words, the sum of the off-diagonal elements of $\mathbf{L}^s$ is smaller than the associated diagonal elements.

To circumvent this issue, we focus on estimating $\tilde{\mathbf{L}}_B^s:=\mathrm{diag}(\mathbf{W}_B^s\mathbf{1})-\mathbf{W}_B^s$ instead of $\mathbf{L}_B^s$, where $\mathbf{W}_B^s$ represents the adjacency matrix that captures the weights of edges between the observable nodes. In this sense, the constraint $\tilde{\mathbf{L}}_B^s\mathbf{1}=\mathbf{0}$ is established, ensuring that $\tilde{\mathbf{L}}_B^s$ is a proper combinatorial Laplacian.
The relationship between $\mathbf{L}_B^s$ and $\tilde{\mathbf{L}}_B^s$ is given by $\tilde{\mathbf{L}}_B^s=\mathbf{L}_B^s-\tilde{\mathbf{D}}_{B}^s$, where $\tilde{\mathbf{D}}_{B}^s$ is a $B\times B$ degree matrix containing the number of edges between observable and hidden nodes. Consequently, the total variation of graph signals, as expressed in (\ref{ex4}), can be replaced by 
\begin{align}\label{ex5}
	\mathrm{tr}(\hat{\mathbf{C}}\mathbf{L}^s)
	&=\mathrm{tr}(\hat{\mathbf{C}}_{B}\tilde{\mathbf{L}}_{B}^s)+\mathrm{tr}(\hat{\mathbf{C}}_B\tilde{\mathbf{D}}_{B}^s)+2\mathrm{tr}(\mathbf{P}^s)+\mathrm{tr}(\mathbf{R}^s) \notag\\
	&=\mathrm{tr}(\hat{\mathbf{C}}_{B}\tilde{\mathbf{L}}_{B}^s)+2\mathrm{tr}(\tilde{\mathbf{P}}^s)+\mathrm{tr}(\mathbf{R}^s),
\end{align}
where $\tilde{\mathbf{P}}^s=\hat{\mathbf{C}}_B\tilde{\mathbf{D}}_{B}^s/2+\mathbf{P}^s$. Exploiting the fact that $\tilde{\mathbf{D}}_{B}^s$ is a low-rank matrix in sparse graph, we can infer that $\tilde{\mathbf{P}}^s$ exhibits the same low-rank structure as $\mathbf{P}^s$. 

Note that the matrix $\tilde{\mathbf{P}}^s$ is not only inseparable from the product of the matrices $\hat{\mathbf{C}}_B$ and $\tilde{\mathbf{D}}_{B}^s$, but also related to the matrix $\mathbf{P}^s$. Beside, the diagonal elements of $\tilde{\mathbf{D}}_{B}^s$ are sparse, leading to the product $\hat{\mathbf{C}}_B\tilde{\mathbf{D}}_{B}^s$ having several zero columns. Following with the assumption \textbf{AS1}, it is evident that $\tilde{\mathbf{P}}^s$ presents a column-sparse structure. 
We incorporate the group Lasso penalty to encourage the desired column-sparse structure of $\tilde{\mathbf{P}}^s$. Consequently, based on the assumptions \textbf{AS1}-\textbf{AS3} in Subsection \ref{sec3.1}, a column-sparsity regularization framework for signed graph learning with hidden nodes is formulated as
\begin{align}\label{ex7}
	\min\limits_{\{\tilde{\mathbf{L}}_{B}^{s},\tilde{\mathbf{P}}^{s},\mathbf{R}^{s}\}_{s\in\{+,-\}}}\quad
	&\mathrm{tr}(\hat{\mathbf{C}}_{B}\tilde{\mathbf{L}}_{B}^{+})+2\mathrm{tr}(\tilde{\mathbf{P}}^{+})+\mathrm{tr}(\mathbf{R}^{+})-\mathrm{tr}(\hat{\mathbf{C}}_B\tilde{\mathbf{L}}_{B}^{-})-2\mathrm{tr}(\tilde{\mathbf{P}}^{-})-\mathrm{tr}(\mathbf{R}^{-})\notag\\
	&+\sum_{s\in\{+,-\}}\alpha_{s}\Vert\tilde{\mathbf{L}}_{B}^{s}\Vert_{F}^2+\sigma_{s}\Vert\tilde{\mathbf{P}}^{s}\Vert_{2,1}\notag\\
	\mathrm{s.t.}\quad
	&\mathrm{tr}(\hat{\mathbf{C}}_{B}\tilde{\mathbf{L}}_{B}^{s})+2\mathrm{tr}(\tilde{\mathbf{P}}^{s})+\mathrm{tr}(\mathbf{R}^{s})\ge0,\notag\\		
	&\mathrm{tr}(\mathbf{R}^{s})\ge 0,\notag\\
	&\mathrm{tr}(\tilde{\mathbf{L}}_{B}^{s})=B,\notag\\
	& \tilde{\mathbf{L}}_{B}^s\in\mathcal{L}\ \text{and}\ \tilde{\mathbf{L}}_{B}^s\in\mathcal{U} ,\forall s\in\{+,-\},
\end{align}
where $\alpha_s$ and $\sigma_s$  are positive regularization parameters that control the trade-off between the regularizer. The Frobenius norm $\Vert\tilde{\mathbf{L}}_B^s\Vert_F^2$ serves as a penalty term to control the density of the learned signed graph. The regularization $\Vert\tilde{\mathbf{P}}^{s}\Vert_{2,1}$ is regarded as the group Lasso penalty. The first constraint in (\ref{ex7}) ensures that the total variations of graph signals on the two unsigned graphs are non-negative. The second constraint captures the fact that the matrix $\mathbf{R}^s$ is a product of two positive semi-definite matrices. The remaining constraints follow the same structure as those presented in optimization problem (\ref{ex2}).
 
\section{Proposed Algorithm}\label{sec4}
The optimization problem in (\ref{ex7}) is not jointly convex with respect to variables $\tilde{\mathbf{L}}_{B}^{s},\tilde{\mathbf{P}}^{s},\mathbf{R}^{s}$, due to the non-convex constraint $ \tilde{\mathbf{L}}_{B}^s\in\mathcal{U} ,\forall s\in\{+,-\}$. To solve the problem, we propose an algorithm based on the block coordinate descent (BCD) \cite{tse} approach. The proposed algorithm decomposes the original problem into three simpler subproblems. Specifically, we use the alternating direction method of multipliers (ADMM) framework to handle the non-convex optimization subproblem for the $\tilde{\mathbf{L}_B^{s}}$-update, as detailed in Subsection \ref{up1}.
In Subsection \ref{up2} and Subsection \ref{up3}, the two convex optimization subproblems for the $\tilde{\mathbf{P}}^{s}$-update and $\mathbf{R}^{s}$-update are solved using CVX \cite{cvx30}.
We initialize $(\tilde{\mathbf{L}}_B^{s^{(m)}},\tilde{\mathbf{P}}^{s^{(m)}}, \mathbf{R}^{s^{(m)}})$ for $m=0$, where $m\in\{0,\cdots,M-1\}$ represents the iteration index. To
simplify the notations, we drop the superscripts $^{(m)}$ from constant variables. The update steps are detailed below.

\subsection{$\tilde{\mathbf{L}}_B^{s}$-update Step}\label{up1}
Assuming $\tilde{\mathbf{P}}^{s}$ and $\mathbf{R}^{s}$ to be fixed, the optimization problem in (\ref{ex7}) can be restated in terms of $\tilde{\mathbf{L}}_{B}^s$ as
\begin{align}\label{ex8}
	\min\limits_{\tilde{\mathbf{L}}_{B}^{+},\tilde{\mathbf{L}}_{B}^{-}}\quad
	&\mathrm{tr}(\hat{\mathbf{C}}_{B}\tilde{\mathbf{L}}_{B}^{+})-\mathrm{tr}(\hat{\mathbf{C}}_B\tilde{\mathbf{L}}_{B}^{-})
	+\sum_{s\in\{+,-\}}\alpha_{s}\Vert\tilde{\mathbf{L}}_{B}^{s}\Vert_{F}^2\notag\\
	\mathrm{s.t.}\quad
	&\mathrm{tr}(\hat{\mathbf{C}}_{B}\tilde{\mathbf{L}}_{B}^{s})+2\mathrm{tr}(\tilde{\mathbf{P}}^{s})+\mathrm{tr}(\mathbf{R}^{s})\ge0,\notag\\		
	&\mathrm{tr}(\tilde{\mathbf{L}}_{B}^{s})=B,\notag\\
	& \tilde{\mathbf{L}}_{B}^s\in\mathcal{L}\ \text{and}\ \tilde{\mathbf{L}}_{B}^s\in\mathcal{U} ,\forall s\in\{+,-\}.
\end{align} 
Since the matrix $\tilde{\mathbf{L}}_{B}^s$ is symmetric, it is convenient to write the subproblem (\ref{ex8}) into vector form. Instead of working with the full square matrix, we focus on the upper triangular part of $\tilde{\mathbf{L}}_{B}^s$. 
Let $\boldsymbol{\ell}^+=\mathrm{upper}(\tilde{\mathbf{L}}_B^+)$, $\boldsymbol{\ell}^-=\mathrm{upper}(\tilde{\mathbf{L}}_{B}^-)$, $\boldsymbol{c}=\mathrm{upper}(\hat{\mathbf{C}}_{B})$, $\boldsymbol{d}=\mathrm{diag}(\hat{\mathbf{C}}_{B})$, $\boldsymbol{p}^+=\mathrm{diag}(\tilde{\mathbf{P}}^+)$, $\boldsymbol{p}^-=\mathrm{diag}(\tilde{\mathbf{P}}^-)$, $\boldsymbol{r}^+=\mathrm{diag}(\mathbf{R}^+)$ and $\boldsymbol{r}^-=\mathrm{diag}(\mathbf{R}^-)$. Define a matrix $\mathbf{S}\in\mathbb{R}^{B\times B(B-1)/2}$ such that $\mathbf{S}\boldsymbol{a}=\mathbf{A1}$, where the vector $\boldsymbol{a}=\mathrm{upper}(\mathbf{A})$ and $\mathbf{A}$ is a symmetric matrix with zero diagonal entries. By introducing these notations, the optimization problem in (\ref{ex8}) can be reformulated as
\begin{align}\label{ex9}
	\min\limits_{\boldsymbol{\ell}^+\leq 0, \boldsymbol{\ell}^-\leq 0}\quad	&\langle 2\boldsymbol{c}-\mathbf{S}^\top \boldsymbol{d},\boldsymbol{\ell}^+\rangle+\alpha_+\langle (2\mathbf{I}+\mathbf{S}^\top\mathbf{S})\boldsymbol{\ell}^+,\boldsymbol{\ell}^+\rangle\notag\\
	&-\langle 2\boldsymbol{c}-\mathbf{S}^\top \boldsymbol{d} ,\boldsymbol{\ell}^-\rangle+\alpha_-\langle (2\mathbf{I}+\mathbf{S}^\top\mathbf{S})\boldsymbol{\ell}^-,\boldsymbol{\ell}^-\rangle\notag\\
	&-\frac{1}{\eta_+}\log_{}{[\langle 2\boldsymbol{c}-\mathbf{S}^\top \boldsymbol{d}, \boldsymbol{\ell}^+\rangle+2(\mathbf{1}^\top\boldsymbol{p}^+) +\mathbf{1}^\top\boldsymbol{r}^+]}\notag\\
	&-\frac{1}{\eta_-}\log_{}[\langle 2\boldsymbol{c}-\mathbf{S}^\top \boldsymbol{d}, \boldsymbol{\ell}^-\rangle	+2(\mathbf{1}^\top\boldsymbol{p}^-) +\mathbf{1}^\top\boldsymbol{r}^-]\notag\\
	\mathrm{s.t.}\quad
	&\mathbf{1}^\top\boldsymbol{\ell}^+=-B/2, \mathbf{1}^\top\boldsymbol{\ell}^-=-B/2,\notag\\
	&\boldsymbol{\ell}^+\bot\boldsymbol{\ell}^-	,
\end{align}
where $\langle \cdot,\cdot\rangle$ and $\bot$ represent the inner product operator and the vertical operator, respectively.
The first and third terms of the objective function in (\ref{ex9}) correspond to the trace terms in (\ref{ex8}), the second and fourth terms correspond to the Frobenius norm. The total variation constraint is incorporated into the objective function as a penalty term, utilizing two positive parameters $\eta_{+}$ and $\eta_{-}$, which correspond to the last two terms of the objective function in (\ref{ex9}). The last constraint $\boldsymbol{\ell}^+\bot\boldsymbol{\ell}^-$ 
corresponds to the set constraint $\tilde{\mathbf{L}}_{B}^s\in\mathcal{U} ,\forall s\in\{+,-\}$ in (\ref{ex8}). The constraint $\boldsymbol{\ell}^+\bot\boldsymbol{\ell}^-$ along with $\boldsymbol{\ell}^+\leq 0,\ \boldsymbol{\ell}^-\leq 0$ is referred to as a complementarity constraint \cite{sch}, contributing to the non-convexity of the problem in (\ref{ex9}). In \cite{wanl}, it is proved that the ADMM algorithm can guarantee convergence when applied to optimization problems with complementarity constraints, provided certain assumptions are met. To apply the ADMM algorithm effectively, we introduce two auxiliary variables, $\mathbf{v}^+=\boldsymbol{\ell}^+$ and $\mathbf{v}^-=\boldsymbol{\ell}^-$. Thus, we obtain the following reformulation of problem (\ref{ex9})
\begin{align}\label{ex10}
	\min\limits_{\mathbf{v}^+,\mathbf{v}^-,\boldsymbol{\ell}^+, \boldsymbol{\ell}^-}\quad
	& f(\boldsymbol{\ell}^+, \boldsymbol{\ell}^-)+\imath_{\mathcal{D}}(\mathbf{v}^+,\mathbf{v}^-)+\imath_{\mathcal{S}}(\boldsymbol{\ell}^+)+\imath_{\mathcal{S}}(\boldsymbol{\ell}^-)	\notag\\
	\mathrm{s.t.}\quad
	&\mathbf{v}^+-\boldsymbol{\ell}^+=\mathbf{0},\mathbf{v}^--\boldsymbol{\ell}^-=\mathbf{0}.
\end{align}	
Here, $f(\boldsymbol{\ell}^+, \boldsymbol{\ell}^-)$ is the objective function in problem (\ref{ex9}), $\imath_{\mathcal{D}}(\cdot)$ and $\imath_{\mathcal{S}}(\cdot)$ are the indicator functions for the constraint set $\mathcal{D}=\{(\mathbf{v}^+,\mathbf{v}^-):\mathbf{v}^+\leq 0,\mathbf{v}^-\leq 0,\mathbf{v}^+\bot\mathbf{v}^- \}$ and the hyperplane $\mathcal{S}=\{\boldsymbol{\ell}^s:\mathbf{1}^\top\boldsymbol{\ell}^s=-B/2, \forall s\in\{+,-\}\}$, respectively. Letting $h(\boldsymbol{\ell}^+,\boldsymbol{\ell}^-)= f(\boldsymbol{\ell}^+, \boldsymbol{\ell}^-)+\imath_{\mathcal{S}}(\boldsymbol{\ell}^+)+\imath_{\mathcal{S}}(\boldsymbol{\ell}^-)$, the problem (\ref{ex10}) can be reformulated in standard ADMM form
\begin{align}\label{admm}
	\min\limits_{\mathbf{v}^+,\mathbf{v}^-,\boldsymbol{\ell}^+,\boldsymbol{\ell}^-}\quad
	& h(\boldsymbol{\ell}^{+},\boldsymbol{\ell}^{-})+\imath_{\mathcal{D}}(\mathbf{v}^{+},\mathbf{v}^{-})	\notag\\
	\mathrm{s.t.}\quad
	&\mathbf{v}^{+}-\boldsymbol{\ell}^{+}=\mathbf{0}, \mathbf{v}^{-}-\boldsymbol{\ell}^{-}=\mathbf{0}.
\end{align}
The corresponding augmented Lagrangian is given as

\begin{align}\label{ex11}
	\mathcal{L}_{\rho}(\mathbf{v}^{+},\mathbf{v}^{-},\boldsymbol{\ell}^{+},\boldsymbol{\ell}^{-},\boldsymbol{\lambda}_{1},\boldsymbol{\lambda}_{2})&=h(\boldsymbol{\ell}^{+},\boldsymbol{\ell}^{-})+\imath_{\mathcal{D}}(\mathbf{v}^{+},\mathbf{v}^{-})+\boldsymbol{\lambda}_{1}^\top(\mathbf{v}^{+}-\boldsymbol{\ell}^{+})\notag\\
	&\quad+\frac{\rho}{2}\Vert\mathbf{v}^{+}-\boldsymbol{\ell}^{+}\Vert_2^2+\boldsymbol{\lambda}_{2}^\top(\mathbf{v}^{-}-\boldsymbol{\ell}^{-})+\frac{\rho}{2}\Vert\mathbf{v}^{-}-\boldsymbol{\ell}^{-}\Vert_2^2,
\end{align}
where $\boldsymbol{\lambda}_1$ and $\boldsymbol{\lambda}_2$ are dual variables, and $\rho\ge 0$ is the ADMM penalty parameter.

The $t$-th ADMM iteration consists of three updates, where $t$ denotes the iteration number. Fixing ${\boldsymbol{\ell}^{+}}^{(t)}, {\boldsymbol{\ell}^{-}}^{(t)},\boldsymbol{\lambda}_1^{(t)}\ \mathrm{and}\ \boldsymbol{\lambda}_2^{(t)}$, the update of $(\mathbf{v}^+,\mathbf{v}^-)$ is
\begin{align}\label{ex12}
	&({\mathbf{v}^{+}}^{(t+1)},{\mathbf{v}^{-}}^{(t+1)})=\min\limits_{\mathbf{v}^+,\mathbf{v}^-}\quad\imath_{\mathcal{D}}(\mathbf{v}^+,\mathbf{v}^-)+\frac{\rho}{2}\big\Vert\mathbf{v}^{+}-{\boldsymbol{\ell}^{+}}^{(t)}+\frac{\boldsymbol{\lambda}_1^{(t)}}{\rho}\big\Vert_2^2+\frac{\rho}{2}\big\Vert\mathbf{v}^{-}-{\boldsymbol{\ell}^{-}}^{(t)}+\frac{\boldsymbol{\lambda}_2^{(t)}}{\rho}\big\Vert_2^2,
\end{align}
which has a closed-form solution
\begin{align}\label{ex13}
	({\mathbf{v}^{+}}^{(t+1)},{\mathbf{v}^{-}}^{(t+1)})=\mathrm{\Pi}_\mathcal{D}\left[\left(\Big({\boldsymbol{\ell}^{+}}^{(t)}-\frac{\boldsymbol{\lambda}_1^{(t)}}{\rho}\Big)^\top, \Big({\boldsymbol{\ell}^{-}}^{(t)}-\frac{\boldsymbol{\lambda}_2^{(t)}}{\rho}\Big)^\top\right)^\top\right],
\end{align}
where $\mathrm{\Pi}_\mathcal{D}[\cdot]$ is the projection operator on the set $\mathcal{D}$.

Fixing $(\boldsymbol{\lambda}_1^{(t)}$, $\boldsymbol{\lambda}_2^{(t)})$ and leveraging the newly updated values from the previous step $({\mathbf{v}^{+}}^{(t+1)}$, ${\mathbf{v}^{-}}^{(t+1)})$, the update for $(\boldsymbol{\ell}^{+}, \boldsymbol{\ell}^{-})$ is given by
\begin{align}\label{ex14}
	({\boldsymbol{\ell}^{+}}^{(t+1)},{\boldsymbol{\ell}^{-}}^{(t+1)})=
	\min\limits_{\boldsymbol{\ell}^+,\boldsymbol{\ell}^-}\quad &f(\boldsymbol{\ell}^+, \boldsymbol{\ell}^-)+\imath_{\mathcal{S}}(\boldsymbol{\ell}^+)+\imath_{\mathcal{S}}(\boldsymbol{\ell}^-)+\frac{\rho}{2}\big\Vert{\mathbf{v}^{+}}^{(t+1)}-\boldsymbol{\ell}^{+}+\frac{\boldsymbol{\lambda}_1^{(t)}}{\rho}\big\Vert_2^2\notag\\
	&+\frac{\rho}{2}\big\Vert{\mathbf{v}^{-}}^{(t+1)}-\boldsymbol{\ell}^{-}+\frac{\boldsymbol{\lambda}_2^{(t)}}{\rho}\big\Vert_2^2.
\end{align}
\begin{algorithm}[!t]
	\caption{ADMM Algorithm for Solving Subproblem (\ref{ex8}) }
	\label{alg:alg1}
	\renewcommand{\algorithmicrequire}{\textbf{Input:}}
	\renewcommand{\algorithmicensure}{\textbf{Output:}}
	\renewcommand{\algorithmicrequire}{\textbf{Input:}}
	\begin{algorithmic}[1]  
		\STATE Initialize ${\boldsymbol{\ell}^{+}}^{(0)}$, ${\boldsymbol{\ell}^{-}}^{(0)}$, $\boldsymbol{\lambda}_1^{(0)}$ and $\boldsymbol{\lambda}_2^{(0)}$ 
		\FOR {$t=0, 1, 2,
			\cdots, T-1$}
		\STATE Update ${\mathbf{v}^{+}}^{(t+1)}$, ${\mathbf{v}^{-}}^{(t+1)}$ according to (\ref{ex13})
		\STATE Update ${\boldsymbol{\ell}^{+}}^{(t+1)}$, ${\boldsymbol{\ell}^{-}}^{(t+1)}$ according to (\ref{ex15})
		\STATE Update $\boldsymbol{\lambda}_1^{(t+1)}$, $\boldsymbol{\lambda}_2^{(t+1)}$ according to (\ref{ex16})
		\ENDFOR
		\STATE  Convert ${\boldsymbol{\ell}^{+}}^{(t+1)}$, ${\boldsymbol{\ell}^{-}}^{(t+1)}$ into matrix form and 
		update $\tilde{\mathbf{L}}_{B}^{{+}^{(t+1)}}$, $\tilde{\mathbf{L}}_{B}^{{-}^{(t+1)}}$
	\end{algorithmic}\label{alg1}
\end{algorithm}
The updates of $\boldsymbol{\ell}^{+},\boldsymbol{\ell}^{-}$ are
\begin{align}\label{ex15}
	{\boldsymbol{\ell}^{+}}^{(t+1)}=\mathrm{\Pi}_\mathcal{S} &\left[{\boldsymbol{\ell}^{+}}^{(t)}-\left(\nabla^2 f({\boldsymbol{\ell}^{+}}^{(t)})+\rho\mathbf{I}\right)^{-1}\left(\nabla f({\boldsymbol{\ell}^{+}}^{(t)})\right.\right.\left.\left.-\rho {\mathbf{v}^{+}}^{(t+1)}+\rho	{\boldsymbol{\ell}^{+}}^{(t)}-\boldsymbol{\lambda}_1^{(t)}\right)\right],\notag\\
	{\boldsymbol{\ell}^{-}}^{(t+1)}=\mathrm{\Pi}_\mathcal{S} &\left[{\boldsymbol{\ell}^{-}}^{(t)}-\left(\nabla^2 f({\boldsymbol{\ell}^{-}}^{(t)})+\rho\mathbf{I}\right)^{-1}\left(\nabla f({\boldsymbol{\ell}^{-}}^{(t)})\right.\right.\left.\left.-\rho {\mathbf{v}^{-}}^{(t+1)}+\rho	{\boldsymbol{\ell}^{-}}^{(t)}-\boldsymbol{\lambda}_2^{(t)}\right)\right],
\end{align}
where $\mathrm{\Pi}_\mathcal{S}[\cdot]$ is the projection operator on the hyperplane $\mathcal{S}$, and $\nabla^2$ and $\nabla$ are represent the Hessian matrix and the gradient operator, respectively. 

Subsequently, the updates of dual variables are
\begin{align}\label{ex16}
	\boldsymbol{\lambda}_1^{(t+1)}=\boldsymbol{\lambda}_1^{(t)}+\rho({\mathbf{v}^{+}}^{(t+1)}-{\boldsymbol{\ell}^{+}}^{(t+1)}),\notag\\
	\boldsymbol{\lambda}_2^{(t+1)}=\boldsymbol{\lambda}_2^{(t)}+\rho({\mathbf{v}^{-}}^{(t+1)}-{\boldsymbol{\ell}^{-}}^{(t+1)}).
\end{align}
The ADMM algorithm for updating $\tilde{\mathbf{L}}_B^{s}$ is placed in Algorithm \ref{alg1}. 

\subsection{$\tilde{\mathbf{P}}^{s}$-update Step}\label{up2}

Given the $\mathbf{R}^s$ and leveraging the estimate $\tilde{\mathbf{L}}_{B}^s$ from the $\tilde{\mathbf{L}}_B^{s}$-update step, we solve the following optimization subproblem with respect to $\tilde{\mathbf{P}}^s$
\begin{align}\label{ex17}
	\min\limits_{\tilde{\mathbf{P}}^+,\tilde{\mathbf{P}}^-}\quad
	&2\mathrm{tr}(\tilde{\mathbf{P}}^{+})-2\mathrm{tr}(\tilde{\mathbf{P}}^{-})+\sum_{s\in\{+,-\}}\sigma_{s}\Vert\tilde{\mathbf{P}}^{s}\Vert_{2,1}\notag\\
	\mathrm{s.t.}\quad
	&\mathrm{tr}(\hat{\mathbf{C}}_{B}\tilde{\mathbf{L}}_{B}^{s})+2\mathrm{tr}(\tilde{\mathbf{P}}^{s})+\mathrm{tr}(\mathbf{R}^{s})\ge0, \forall s\in\{+,-\}.
\end{align}
\subsection{$\mathbf{R}^s$-update Step}\label{up3}
According to the estimate $\tilde{\mathbf{L}}_{B}^s$ and $\tilde{\mathbf{P}}^s$, we minimize the following convex optimization subproblem to update $\mathbf{R}^s$
\begin{align}\label{ex18}
	\min\limits_{\mathbf{R}^+,\mathbf{R}^-}\quad
	&\mathrm{tr}(\mathbf{R}^s)\notag\\
	\mathrm{s.t.}\quad
	&\mathrm{tr}(\hat{\mathbf{C}}_{B}\tilde{\mathbf{L}}_{B}^{s})+2\mathrm{tr}(\tilde{\mathbf{P}}^{s})+\mathrm{tr}(\mathbf{R}^{s})\ge0,\notag\\
	&\mathrm{tr}(\mathbf{R}^s)\ge 0, \forall s\in\{+,-\}.
\end{align}
We summarize the proposed SGL-HNCS method in Algorithm \ref{alg2}.

\section{Convergence and Complexity Analysis} \label{sec5}
In this section, we first provide several supporting lemmas and analyze the convergence of Algorithm \ref{alg1} and Algorithm \ref{alg2}. Then, we analyze the complexity of the proposed SGL-HNCS algorithm.
\begin{algorithm}[!t]
	\caption{The SGL-HNCS (proposed) Algorithm }
	\label{alg:alg2}
	\renewcommand{\algorithmicrequire}{\textbf{Input:}}
	\renewcommand{\algorithmicensure}{\textbf{Output:}}
	\renewcommand{\algorithmicrequire}{\textbf{Input:}}
	\begin{algorithmic}[1]  
		\REQUIRE Original graph signals $\mathbf{X}\in\mathbb{R}^{N\times K}$, number of hidden nodes $H$, number of iterations $M$, parameters $\alpha_s$, $\theta_s$, $\forall s\in\{+,-\}$, $\eta_+$, $\eta_-$, and $\rho$  
		\ENSURE Graph Laplacians $\tilde{\mathbf{L}}_{B}^{+}$, $\tilde{\mathbf{L}}_{B}^{-}$  
		\STATE Initial $\tilde{\mathbf{L}}_{B}^{{+}^{(0)}}$, $\tilde{\mathbf{L}}_{B}^{{-}^{(0)}}$, $\tilde{\mathbf{P}}^{{+}^{(0)}}$, $\tilde{\mathbf{P}}^{{-}^{(0)}}$, $\mathbf{R}^{{+}^{(0)}}$  and $\mathbf{R}^{{-}^{(0)}}$
		\STATE Select the set of hidden nodes $\mathcal{H}=\{1,\cdots,H\}$ by random selection
		\STATE Select the remaining $B=N-H$ rows of original graph  signals as the observed graph signals $\mathbf{X}_B$
		\STATE Compute the sample covariance matrix
		{$\hat{\mathbf{C}}_{B}$}
		\FOR{$m=0,1, 2,\cdots ,M-1$}
		\STATE Update $\tilde{\mathbf{L}}_{B}^{{+}^{(m+1)}}$ and $\tilde{\mathbf{L}}_{B}^{{-}^{(m+1)}}$ via Algorithm \ref{alg1}
		\STATE Update $\tilde{\mathbf{P}}^{+^{(m+1)}}$ and $\tilde{\mathbf{P}}^{-^{(m+1)}}$ by solving (\ref{ex17})
		\STATE Update $\mathbf{R}^{+^{(m+1)}}$ and $\mathbf{R}^{-^{(m+1)}}$ by solving (\ref{ex18})
		\ENDFOR	
		\STATE $\tilde{\mathbf{L}}_B^+=\tilde{\mathbf{L}}_B^{+^{(m+1)}}$, $\tilde{\mathbf{L}}_B^-=\tilde{\mathbf{L}}_B^{-^{(m+1)}}$	
	\end{algorithmic}\label{alg2}
\end{algorithm}
\subsection{Convergence Analysis of Algorithm \ref{alg1}}
To guide the convergence of Algorithm \ref{alg1}, we introduce the following lemmas.

\newtheorem{lemma}{Lemma}
\newproof{pf}{Proof}

\begin{lemma}\label{lem1}
	The objective function $h(\boldsymbol{\ell}^{+},\boldsymbol{\ell}^{-})+\imath_{\mathcal{D}}(\mathbf{v}^{+},\mathbf{v}^{-})$ in (\ref{ex10}) is coercive over the constraint set $\mathcal{F}:=\{(\mathbf{v}^{+},\mathbf{v}^{-},\boldsymbol{\ell}^{+},\boldsymbol{\ell}^{-}): \mathbf{v}^{+}-\boldsymbol{\ell}^{+}=\mathbf{0}, \mathbf{v}^{-}-\boldsymbol{\ell}^{-}=\mathbf{0}\}$.
\end{lemma}
\begin{proof}
	Refer to \ref{app1}.
\end{proof}

\begin{lemma}\label{sub_Lc}
	The solutions to the two problems (\ref{ex12}) and (\ref{ex14}) of the optimization problem (\ref{ex10}) are Lipschitz continuous.
\end{lemma}
\begin{proof}
	Refer to \ref{app2}.
\end{proof}

\begin{lemma}\label{Lip}
	The function $h(\boldsymbol{\ell}^{+},\boldsymbol{\ell}^{-})$ is Lipschitz differentiable.
\end{lemma}
\begin{proof}
	Refer to \ref{app3}.
\end{proof}

\newtheorem{theorem}{Theorem}
\begin{theorem}\label{the}
	For a sufficiently large penalty parameter $\rho$, the sequence $({\mathbf{v}^{+}}^{(t)},{\mathbf{v}^{-}}^{(t)},{\boldsymbol{\ell}^{+}}^{(t)},{\boldsymbol{\ell}^{-}}^{(t)},	\boldsymbol{\lambda}_1^{(t)}, \\ \boldsymbol{\lambda}_2^{(t)})$ generated by Algorithm \ref{alg1} has limit points, and each limit point is stationary point of the augmented Lagrangian $\mathcal{L}_{\rho}$, i.e., $0\in\partial\mathcal{L}_{\rho}(({\mathbf{v}^{+}}^{*},{\mathbf{v}^{-}}^{*},{\boldsymbol{\ell}^{+}}^{*},{\boldsymbol{\ell}^{-}}^{*},	\boldsymbol{\lambda}_1^{*},	\boldsymbol{\lambda}_2^{*})$.
\end{theorem}
\begin{proof}
	The convergence result stated in Theorem \ref{the} relies on the framework established in \cite{wanl}. To be precise, the assumptions \textit{A2} and \textit{A4} in \cite{wanl} are obviously satisfied in our case (see \cite{wanl} for a detailed proof). Based on Lemma \ref{lem1}-\ref{sub_Lc}, we demonstrate that all the assumptions in \cite{wanl} hold in our context. Hence, we complete the proof.
\end{proof}

\subsection{Convergence Analysis of Algorithm \ref{alg2}}
To establish the convergence property of Algorithm \ref{alg2}, we first introduce some definitions. Let $\phi(\mathbf{z})$ denotes the objective function in (\ref{ex7}), with $\mathbf{z}:=[\mathbf{z}_1^\top, \mathbf{z}_2^\top, \mathbf{z}_3^\top]^\top$ and $\mathbf{z}_1^\top:=(\mathrm{vec}(\tilde{\mathbf{L}}_{O}^{+}),\mathrm{vec}(\tilde{\mathbf{L}}_{O}^{-}))$, $\mathbf{z}_2^\top:=(\mathrm{vec}(\tilde{\mathbf{P}}^+),\mathrm{vec}(\tilde{\mathbf{P}}^-))$, $\mathbf{z}_3^\top:=(\mathrm{vec}(\mathbf{R}^+),\mathrm{vec}(\mathbf{R}^-))$.
Let $\mathcal{Z}^*$ represents the set of stationary points of $\phi(\mathbf{z})$ and  $\mathbf{z}^{(m)}=[(\mathbf{z}_1^{(m)})^\top,(\mathbf{z}_2^{(m)})^\top,(\mathbf{z}_3^{(m)})^\top ]^\top$ denote the solution generated by running three optimization subproblems $m$ times alternately.
Then, the following supporting lemma is introduced for our algorithm analysis.

\begin{lemma}\label{lem5}
	The function $\phi(\mathbf{z})$ is regular at every point in $\mathcal{Z}^*$.
\end{lemma}
\begin{proof}
	Refer to \ref{app4}.
\end{proof}

The following theorem formally states that the Algorithm \ref{alg:alg2} converges to a stationary point.

\begin{theorem}\label{the2}
	With the previous definition in place, the sequence $\mathbf{z}^{(m)}=[(\mathbf{z}_1^{(m)})^\top,\mathbf{z}_2^{(m)})^\top,\mathbf{z}_3^{(m)})^\top ]^\top$ generated by Algorithm \ref{alg2} converges to a stationary point in $\mathcal{Z}^*$ as $m\rightarrow\infty$. 	 
\end{theorem}

\begin{proof}
	To demonstrate the validity of Theorem \ref{the2}, we leverage the results in \cite{tse}. However, it should be noted that the original objective function contains non-convex orthogonal constraint $\boldsymbol{\ell}^+\bot\boldsymbol{\ell}^-$, which only appear in  $\tilde{\mathbf{L}}_B^{s}$-update step. A pioneering work in \cite{fu} proved that the convergence of the general BCD algorithm can still be guaranteed when orthogonality constraints are considered. Therefore, we only need to establish that the proposed algorithm satisfies the original conditions identified in \cite{tse}.
	
	The function $\phi(\mathbf{z})$ is continuous in $\tilde{\mathbf{L}}_{B}^{+}$, $\tilde{\mathbf{L}}_{B}^{-}$ $\tilde{\mathbf{P}}^+$, $\tilde{\mathbf{P}}^-$, $\mathbf{R}^+$ and $\mathbf{R}^-$, ensuring that the level set $\mathcal{Z}^{0}=\{\mathbf{z}: \phi(\mathbf{z})\leq \phi(\mathbf{z}^{(0)})\}$ is closed. In addition, the domain of function $\phi(\mathbf{z})$ is defined as $\textbf{dom} \ \phi=\{\mathbf{z} \mid \phi(\mathbf{z})\textless +\infty\}$, which is bounded. The closedness and boundedness of the level set $\mathcal{Z}^{0}$ make it compact.
	Based on the Lemma \ref{lem5} and the above verifications, we can conclude that the conditions specified in \cite{tse} are satisfied. By invoking the convergence results in \cite{tse} and \cite{fu}, we conclude that the solution of our algorithm converges to a stationary point.
\end{proof}

\subsection{Complexity Analysis}
To analyze the complexity of Algorithm \ref{alg2}, we first provide the analysis of the ADMM-based algorithm. When updating $(\boldsymbol{\ell}^{+},\boldsymbol{\ell}^{-})$, the most time-consuming operations are the calculation of $\mathbf{S}^\top\mathbf{S}$ and $\left(\nabla^2 f({\boldsymbol{\ell}^{+}},{\boldsymbol{\ell}^{-}})+\rho\mathbf{I}\right)^{-1}\nabla f({\boldsymbol{\ell}^{+}}, {\boldsymbol{\ell}^{-}})$, which have $\mathcal{O}(p^2B+p^3)$ complexity, where $p=B(B-1)/2$. When updating $({\mathbf{v}^{+}},{\mathbf{v}^{-}})$ and 	$(\boldsymbol{\lambda}_1,\boldsymbol{\lambda}_2)$, all operations are element-wise, resulting in a complexity of  $\mathcal{O}(p)$. Thus, the overall complexity of the ADMM-based algorithm is $\mathcal{O}\left(T(p^2B+p^3+p)\right)$, where $T$ is the required number of iterations for the ADMM-based algorithm to converge. Then, the complexity of updating $(\tilde{\mathbf{P}}^+,\tilde{\mathbf{P}}^-)$ and $(\mathbf{R}^+,\mathbf{R}^-)$ using CVX is $\mathcal{O}(B^2+H^2)$. Since $H\ll B$, the contribution of $\mathcal{O}(H^2)$ is negligible. Therefore, the overall complexity of Algorithm \ref{alg2} is $\mathcal{O}\left(M(Tp^2B+Tp^3+Tp)+B^2\right)$, where $M$ represents the number of outer iterations in Algorithm 2.

\section{Numerical Experiments}\label{sec6}
This section presents experimental results for comprehensive evaluations of our proposed graph learning method. Subsection \ref{sec6.1} outlines the experimental settings. Subsection \ref{sec6.2} provides a comparative analysis of our proposed method against several classical methods using synthetic data. Subsection \ref{sec6.3} presents experimental results on real-world data.

\subsection{Experimental Settings}\label{sec6.1}
In this subsection, we introduce the experimental settings, which include the evaluation metrics used and the baseline methods selected for comparison.

\subsubsection{Evaluation Metrics}
To assess the performance of our signed graph learning method, a total of five evaluation metrics are employed. The first metric is the multiclass version of relative error (\text{RelErr}), calculated as the average of $\text{RelErr}^+$ and $\text{RelErr}^-$. 
If $\mathbf{L}^{{+}^{*}}$ is the Laplacian of graph $\mathcal{G}^+$ and $\hat{\mathbf{L}}^+$ is the estimated version, the $\text{RelErr}^+$ is given by $\text{RelErr}^+=\Vert\hat{\mathbf{L}}^+-\mathbf{L}^{{+}^{*}}\Vert_2^2/\Vert\mathbf{L}^{{+}^{*}}\Vert_F^2$. Similarly, the $\text{RelErr}^-$ of graph $\mathcal{G}^-$ is obtained by $\text{RelErr}^-=\Vert\hat{\mathbf{L}}^--\mathbf{L}^{{-}^{*}}\Vert_2^2/\Vert\mathbf{L}^{{-}^{*}}\Vert_F^2$. Multiclass version of $\text{F-score}$ is used as the second metric. 
The $\text{F-score}$ is the average of $\text{F-score}^{+}$ and $\text{F-score}^{-}$, where 
$\text{F-score}^{+}$ and $\text{F-score}^{-}$ are metrics that compare the positive and negative edges of the learned graph to the positive and negative edges of the groundtruth, respectively. The $\text{F-score}$ takes values between 0 and 1, with higher values indicating better performance in capturing the signed graph structure. Additionally, \text{Precision} and \text{Recall} are used to evaluate how effectively the true edge structure of the signed graph is captured. The final evaluation metric, Normalized Mutual Information (\text{NMI}), measures the mutual dependence between the learned edge set and the groundtruth graph.

\subsubsection{Baseline Methods for Comparison}
Several relevant graph learning methods are explored for comparison with our proposed method, i.e., SGL-HNCS. The SGL-HNCS method represents the column-sparsity regularization approach proposed in formulation (\ref{ex7}). The SGL-HNLR method is a variant of the method SGL-HNCS, which replaces the term $\Vert\tilde{\mathbf{P}}^{s}\Vert_{2,1}$ in (\ref{ex7}) with $\Vert\tilde{\mathbf{P}}^{s}\Vert_{*}$ to only account for the low-rank structure of $\tilde{\mathbf{P}}^{s}$, similar to the low-rank scheme GSm-LR described in \cite{Bu27}. The $\text{scSGL}$ method described in \cite{kara18} is a signed graph learning technique that does not account for hidden nodes. Conversely, the $\text{GL}$ method proposed in \cite{15} is an unsigned graph learning technique that ignores the presence of hidden nodes. The GSm-GL method proposed in \cite{Bu27} addresses the same problem as the method $\text{sGL}$ but explicitly considers the presence of hidden nodes.

\subsection{Results on Synthetic Data}\label{sec6.2}

To demonstrate the performance of the proposed method, an artificial dataset is first created. Given a signed graph $\mathcal{G}$ with $N=30$ nodes, we create a type of synthetic graph signals $\mathbf{X}$ based on the assumptions \textbf{AS2} and \textbf{AS3}. The data construction process involves two steps, following the method as described in \cite{15}. In the first step, we generate a random unsigned graph $\mathcal{G}^{'}$ from the $\text{Erd\H{o}s-R{\'e}nyi}$ (ER) model \cite{erd} with an edge connection probability of $p=0.3$. The signed graph $\mathcal{G}$ is obtained by selecting half of the edges from $\mathcal{G}^{'}$ and setting them as negative edges with weights of $-1$, while the remaining edges are set as positive edges with weights of $+1$. In the second step, we generate $K$ graph signals based on the graph Laplacians of the constructed signed graph $\mathcal{G}$. According to the definition of a signed graph, the Laplacian matrices $\mathbf{L}^+$ and $\mathbf{L}^-$ correspond to two unsigned graphs $\mathcal{G}^+$ and $\mathcal{G}^-$, respectively. Let $\mathbf{L}^+=\mathbf{U}\mathbf{\Lambda}\mathbf{U}^\top$ and $\mathbf{L}^-=\mathbf{V}\mathbf{\Sigma}\mathbf{V}^\top$ denote the eigendecomposition of the two unsigned graph Laplacians. The graph signals are generated as $\mathbf{X}=(\mathbf{U}h_1(\mathbf{\Lambda})\mathbf{U}^\top+\mathbf{V}h_2(\mathbf{\Sigma})\mathbf{V}^\top)\mathbf{X}_0+\boldsymbol{\epsilon}$, which are the linear combination of the eigenvectors of $\mathbf{L}^+$ and $\mathbf{L}^-$. The $h_1(\mathbf{\Lambda})$ is a low-pass graph filter defined as $h_1(\mathbf{\Lambda})=\mathbf{\Lambda}^\dagger/\Vert\mathbf{\Lambda}^\dagger\Vert_F$, and $h_2(\mathbf{\Sigma})$ is a high-pass graph filter defined as $h_2(\mathbf{\Sigma})=\mathbf{\Sigma}/\Vert\mathbf{\Sigma}\Vert_F$, where the superscript $^\dagger$ is a pseudo-inverse operator. The columns of $\mathbf{X}_0$ are independent realizations of a multivariate Gaussian distribution $\mathbf{X}_0\sim \mathcal{N}(0,\mathbf{I})$. The $\boldsymbol{\epsilon}$ represents additive white noise following a multivariate Gaussian distribution $\boldsymbol{\epsilon}\sim \mathcal{N}(0,\sigma^2\mathbf{I})$.

Next, several experiments are conducted to investigate the performance of the proposed method and each other baseline methods on synthetic data. Different experimental condition settings are considered, including variations in the number of hidden nodes and the number of graph signals.
\begin{figure*}[!t]
	\centering
	\subfloat[]{\includegraphics[width=0.48\columnwidth]{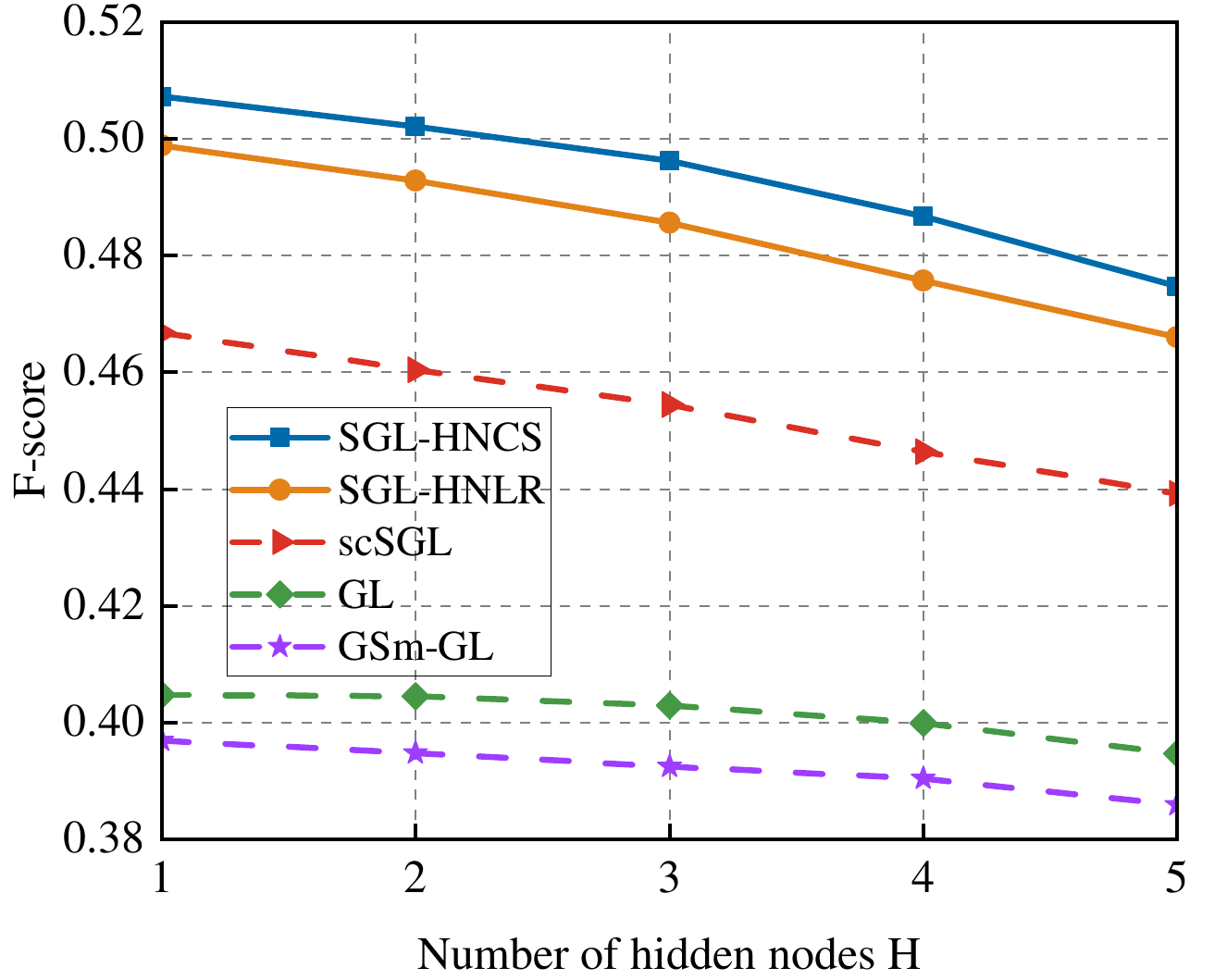}
		\label{f2.1}}
	\subfloat[]{\includegraphics[width=0.48\columnwidth]{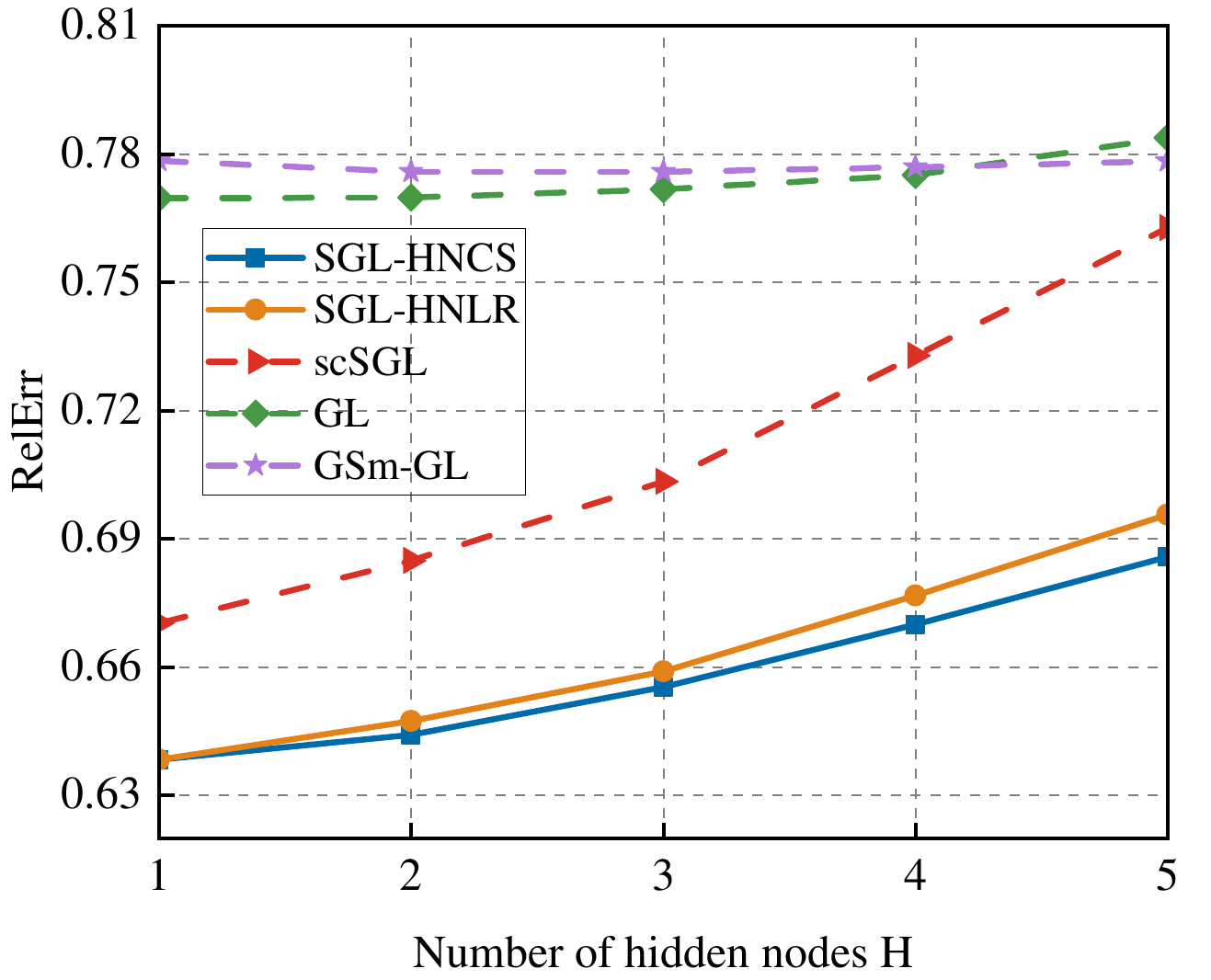}
		\label{f2.2}}
	\caption{Numerical validation of the proposed method under ER graph model with fixed $M$=50. (a) \text{F-score} of the recovered graphs for several methods as the number of hidden nodes increases. (b) \text{RelErr} of recovered graphs as the number of hidden nodes increases for different methods. }
	\label{figr2}
\end{figure*}

\subsubsection{Number of Hidden Nodes}
In the first synthetic experiment, we present how the performance of different methods changes as the number of hidden nodes increases. Specifically, $H$ varies between 1 and 5, and the hidden nodes are randomly selected from the entire set of nodes in the signed graph.

The results of Laplacian matrix estimation, measured by $\text{F-score}$ and $\text{RelErr}$ for different values of $H$ (with a fixed number of graph signals $K=50$), are shown in Fig. \ref{figr2}. It can be seen that for most methods, including $\text{scSGL}$, SGL-HNLR and SGL-HNCS, the $\text{F-score}$ decreases while the $\text{RelErr}$ increases with growing $H$. This is expected, since signed graph learning problem becomes more tractable with an increasing number of hidden nodes. Notably, the method SGL-HNCS and SGL-HNLR , which account for the presence of hidden nodes, outperform the method $\text{scSGL}$. This result indicates the significance of considering hidden nodes in signed graph recovery. Additionally, comparing SGL-HNLR with SGL-HNCS, they exhibit subtle performance differences. This behavior suggests that choosing different regularization constants ($\Vert\cdot\Vert_{*}$ or $\Vert\cdot\Vert_{2,1}$) does not significantly affect the accuracy of the estimated graph. 
Since $\text{GL}$ and GSm-GL learn a graph individually, their performance is not affected by changes in the abscissa variables. $\text{GL}$ shows the worst performance because this method fails to account for hidden nodes. 
	
	\begin{figure}[!t]
		\centering
		{\includegraphics[width=0.49 \columnwidth]{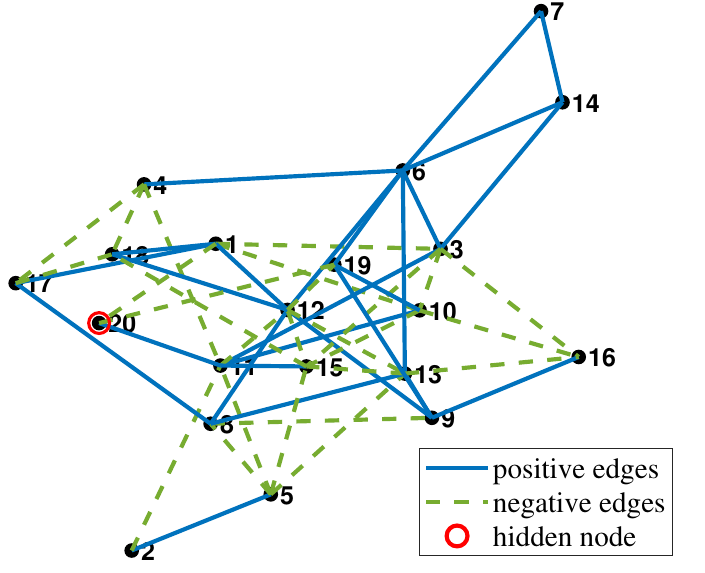}}
		\caption{The designed signed graph based on ER model. }
		\label{graph}
	\end{figure}
	\begin{figure*}[!t]
		\centering
		\subfloat[$\text{Groundtruth}\quad\tilde{\mathbf{L}}_\mathit{B}^{+}$]{\includegraphics[width=0.25\columnwidth]{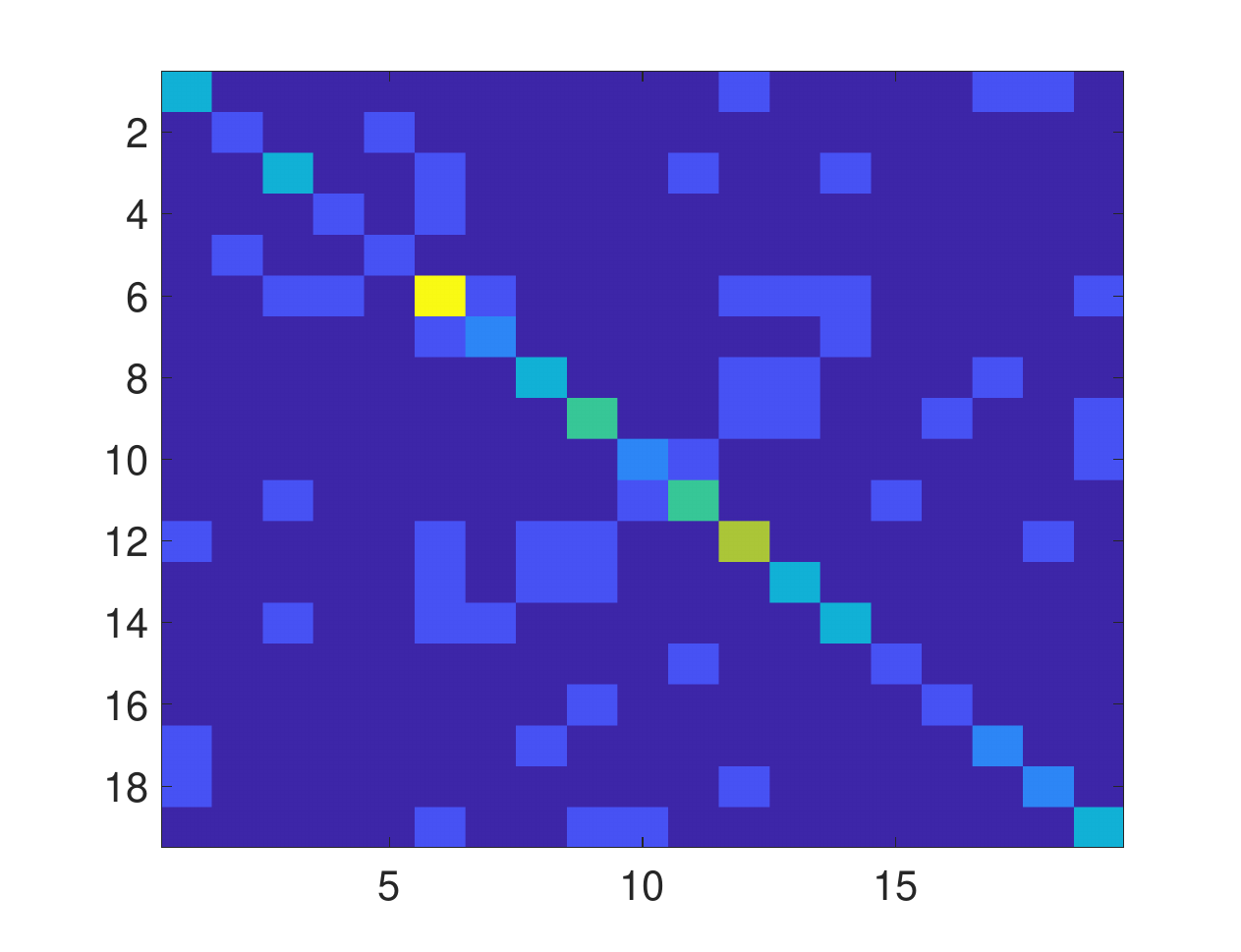}
			\label{f3.2}}
		\subfloat[$\text{SGL-HNLR}\quad\hat{\mathbf{L}}_\mathit{B}^{+}$]{\includegraphics[width=0.25\columnwidth]{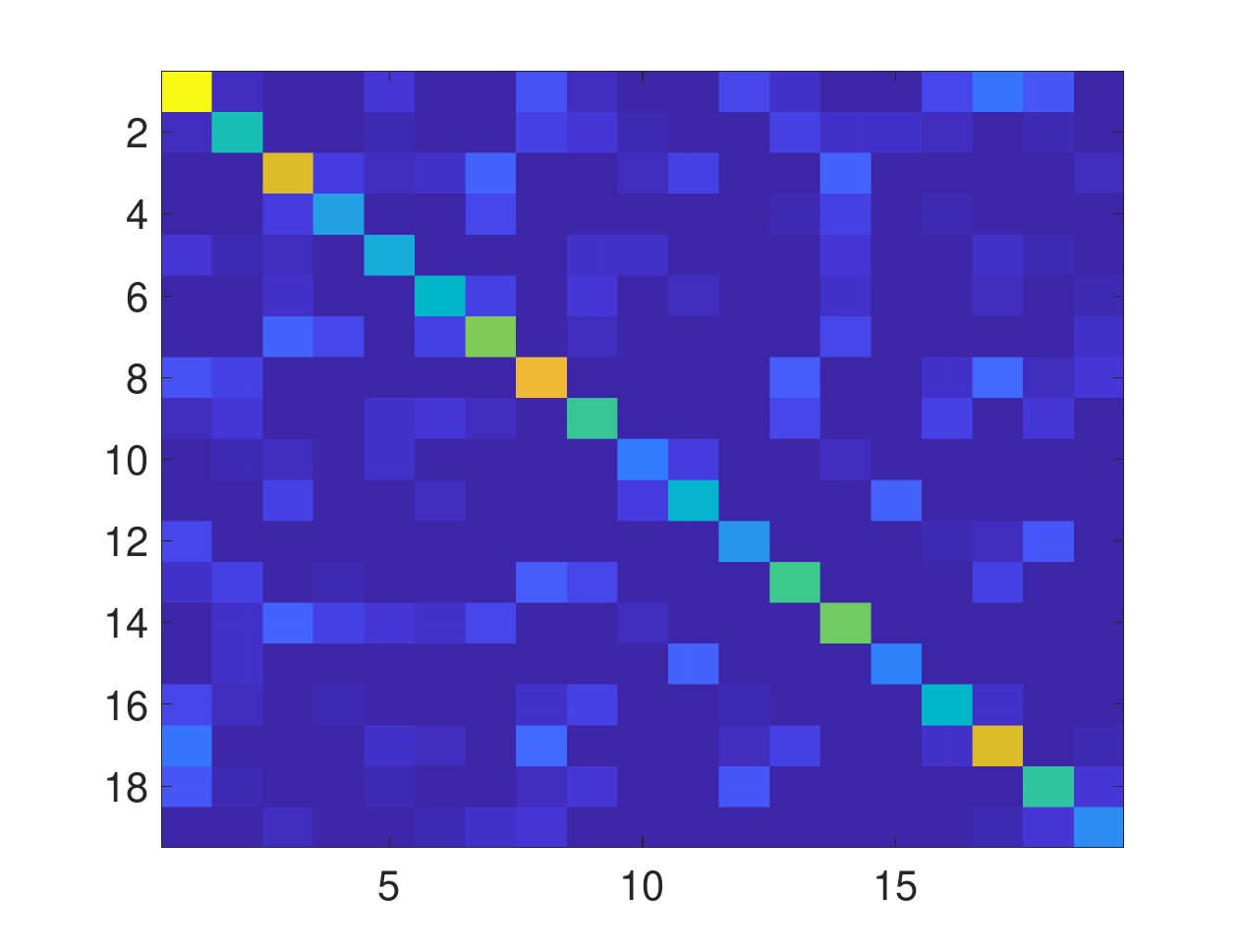}
			\label{f3.3}}
		\subfloat[$\text{SGL-HNCS}\quad\hat{\mathbf{L}}_\mathit{B}^{+}$]{\includegraphics[width=0.25\columnwidth]{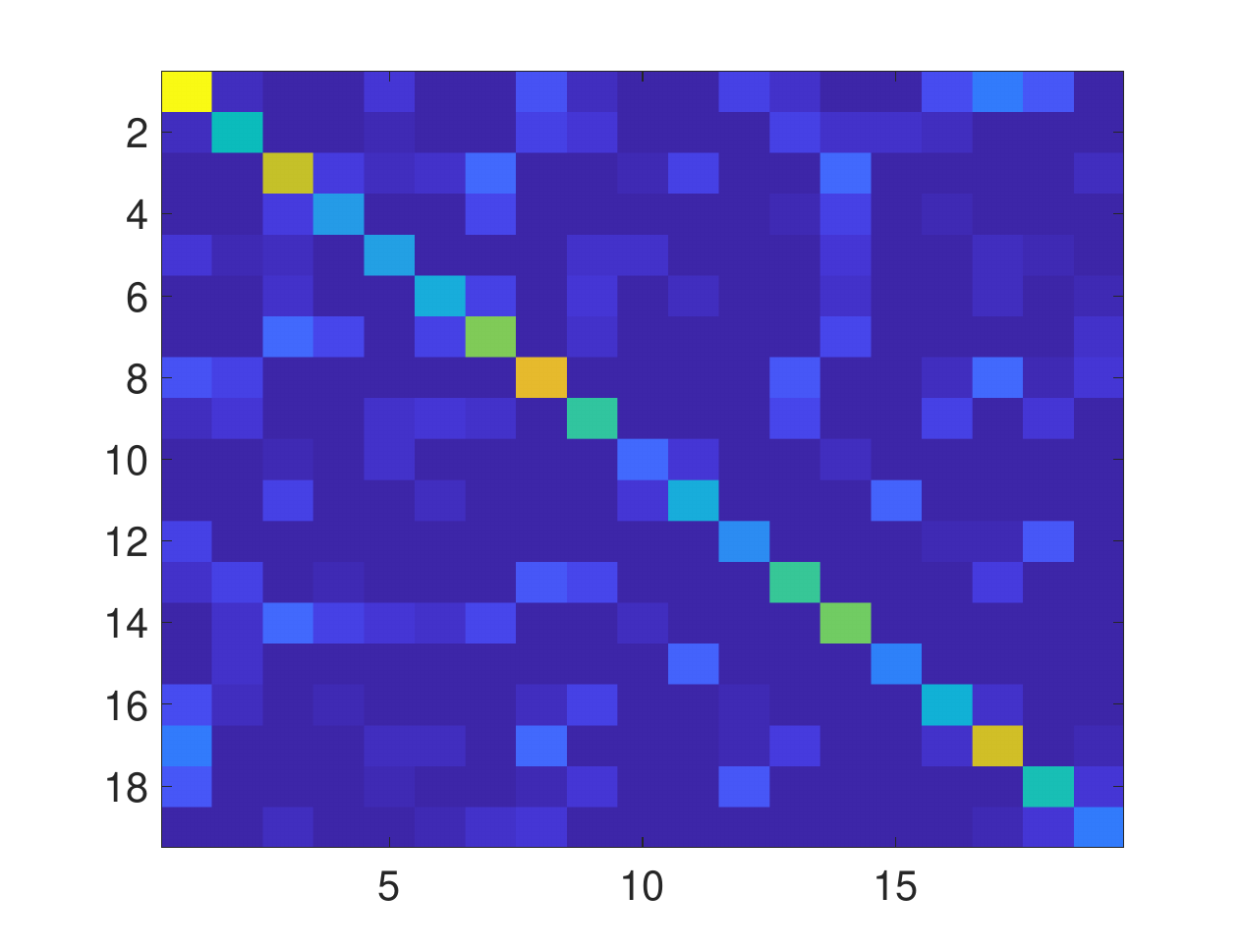}
			\label{f3.4}}
		\subfloat[$\text{scSGL}\quad\hat{\mathbf{L}}_\mathit{B}^{+}$]{\includegraphics[width=0.25\columnwidth]{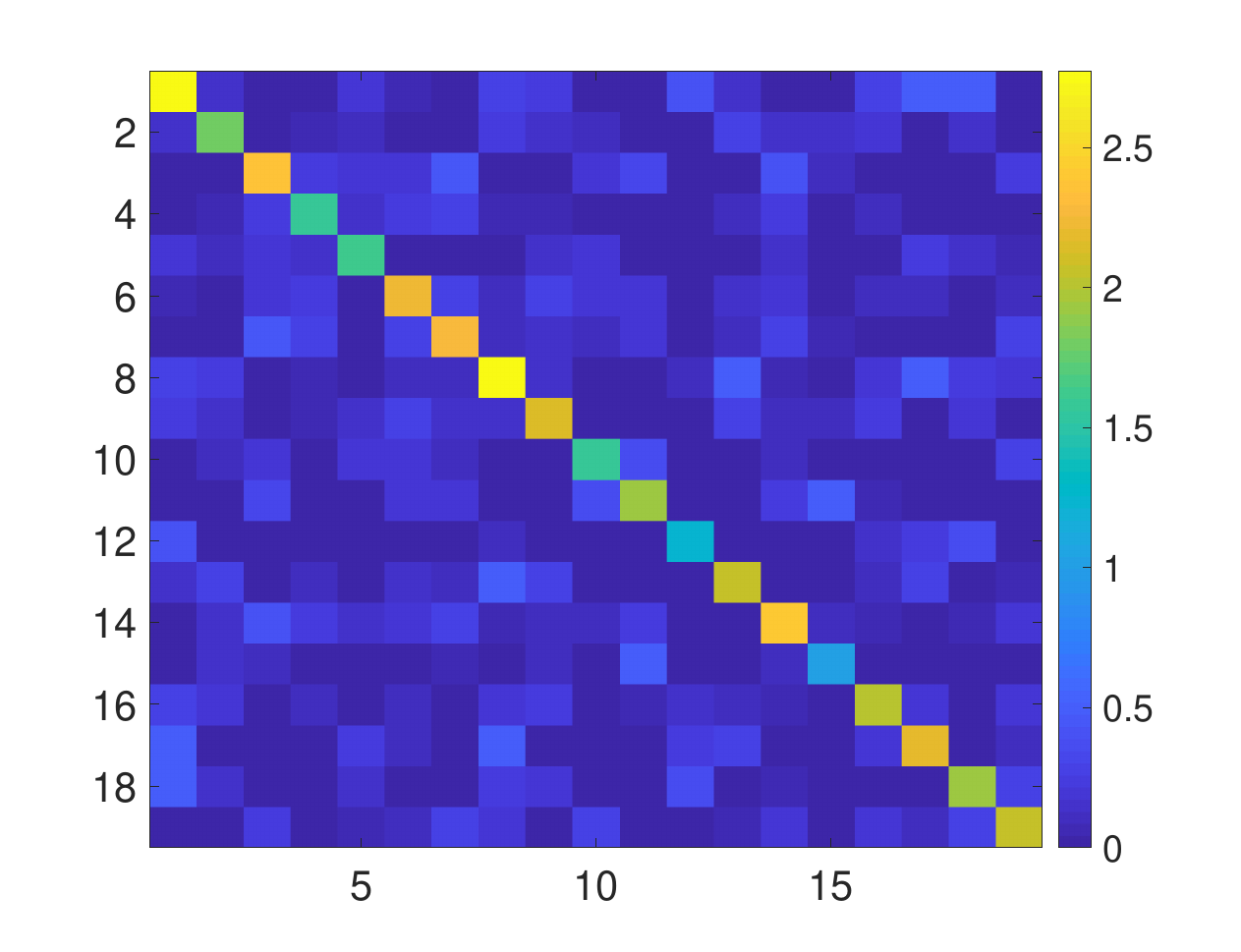}
			\label{f3.5}}\hfil
		\subfloat[$\text{Groundtruth}\quad\tilde{\mathbf{L}}_\mathit{B}^{-}$]{\includegraphics[width=0.25\columnwidth]{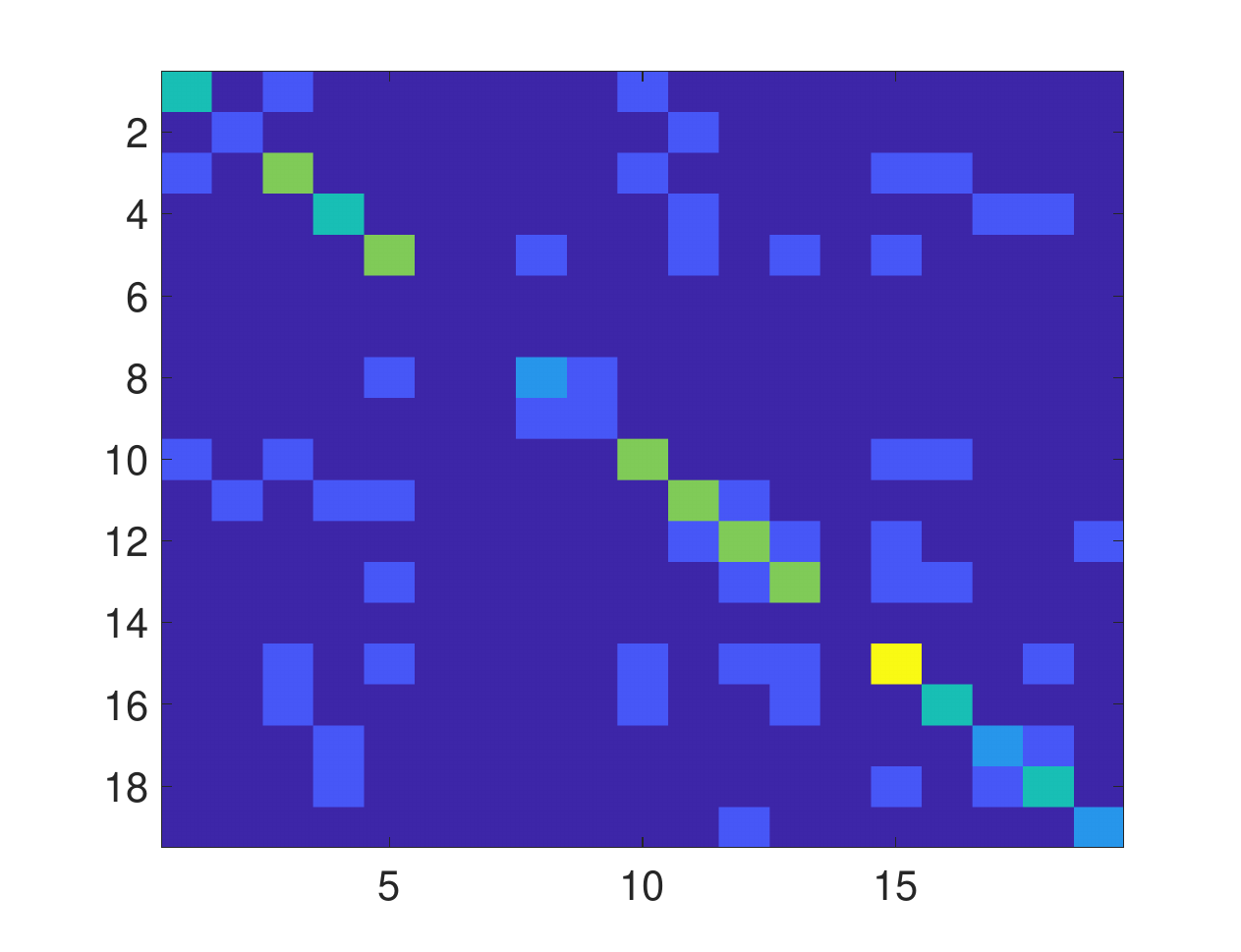}
			\label{f3.6}}
		\subfloat[$\text{SGL-HNLR}\quad\hat{\mathbf{L}}_\mathit{B}^{-}$]{\includegraphics[width=0.25\columnwidth]{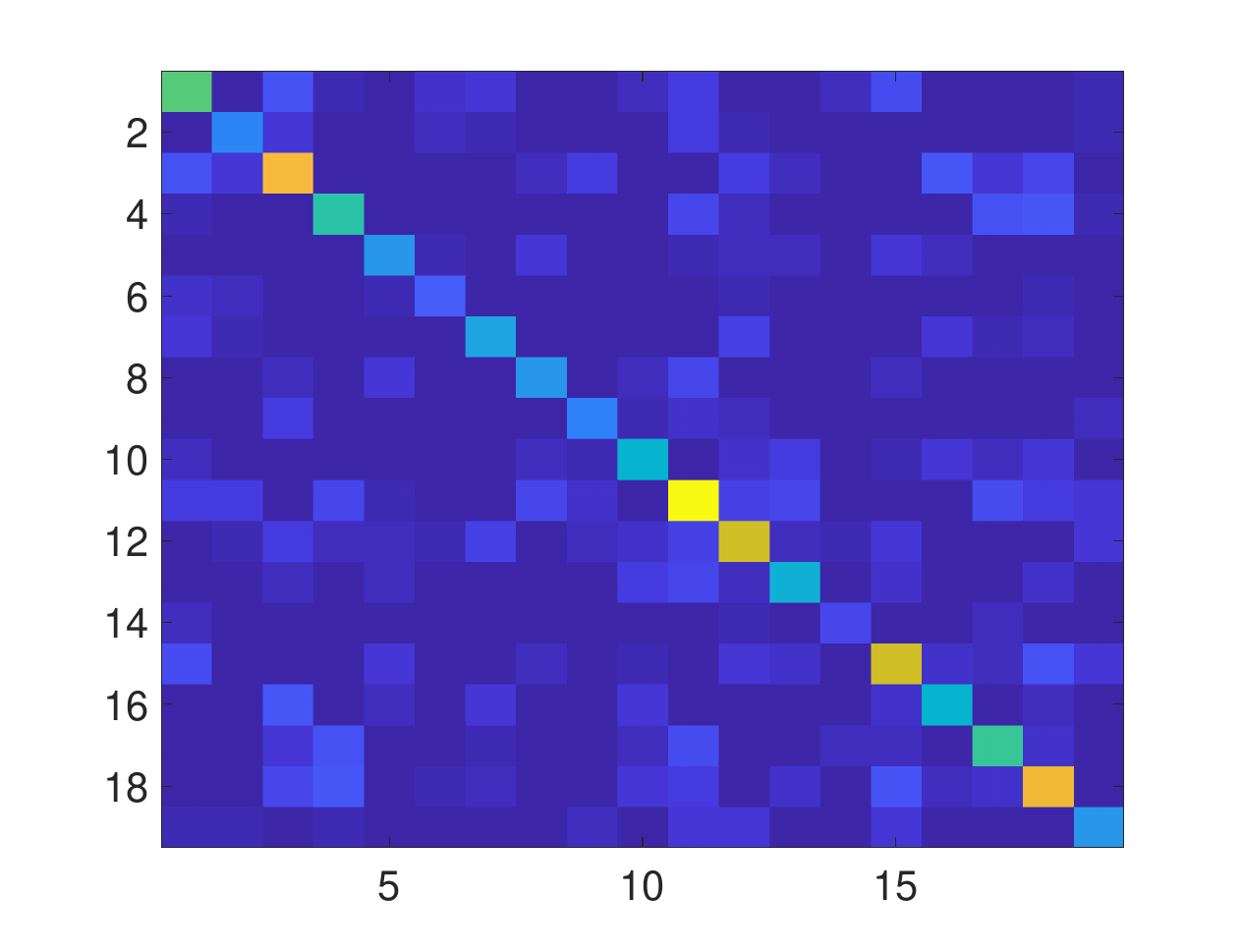}
			\label{f3.7}}
		\subfloat[$\text{SGL-HNCS}\quad\hat{\mathbf{L}}_\mathit{B}^{-}$]{\includegraphics[width=0.25\columnwidth]{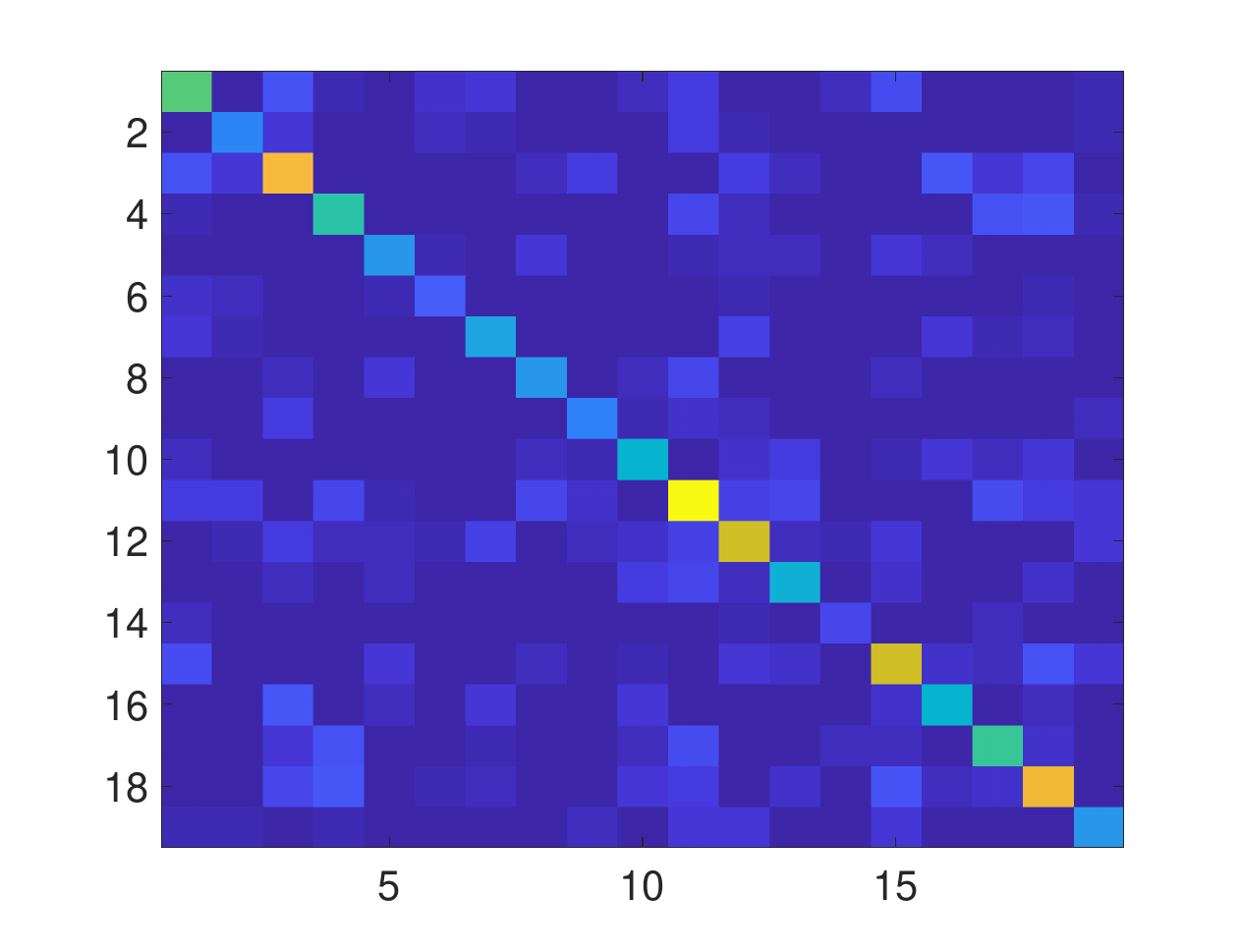}
			\label{f3.8}}
		\subfloat[$\text{scSGL}\quad\hat{\mathbf{L}}_\mathit{B}^{-}$]{\includegraphics[width=0.25\columnwidth]{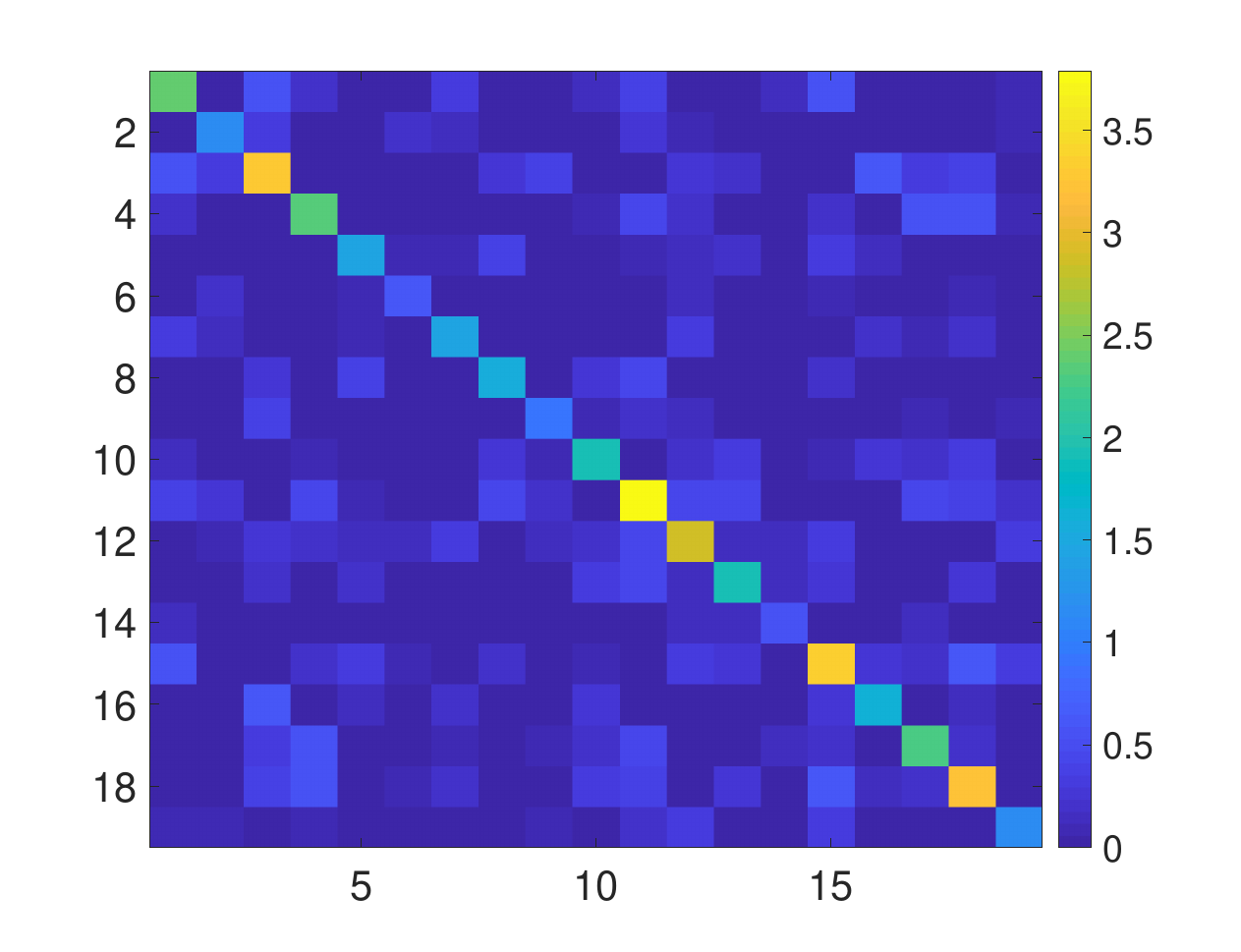}
			\label{f3.9}}
		\caption{The comparison of groundtruth and learned Laplacian matrices. The columns from the left to the right are the groundtruth Laplacians, the Laplacians learned by $\text{SGL-HNLR}$, the Laplacians learned by SGL-HNCS and the Laplacians learned by $\text{scSGL}$. The rows from the top to the bottom are the Laplacians for $\mathcal{G}^+$ and $\mathcal{G}^-$, respectively.  }
		\label{figr3}
	\end{figure*}
	
	\begin{figure*}[!t]
		\centering
		\subfloat[]{\includegraphics[width=0.48\columnwidth]{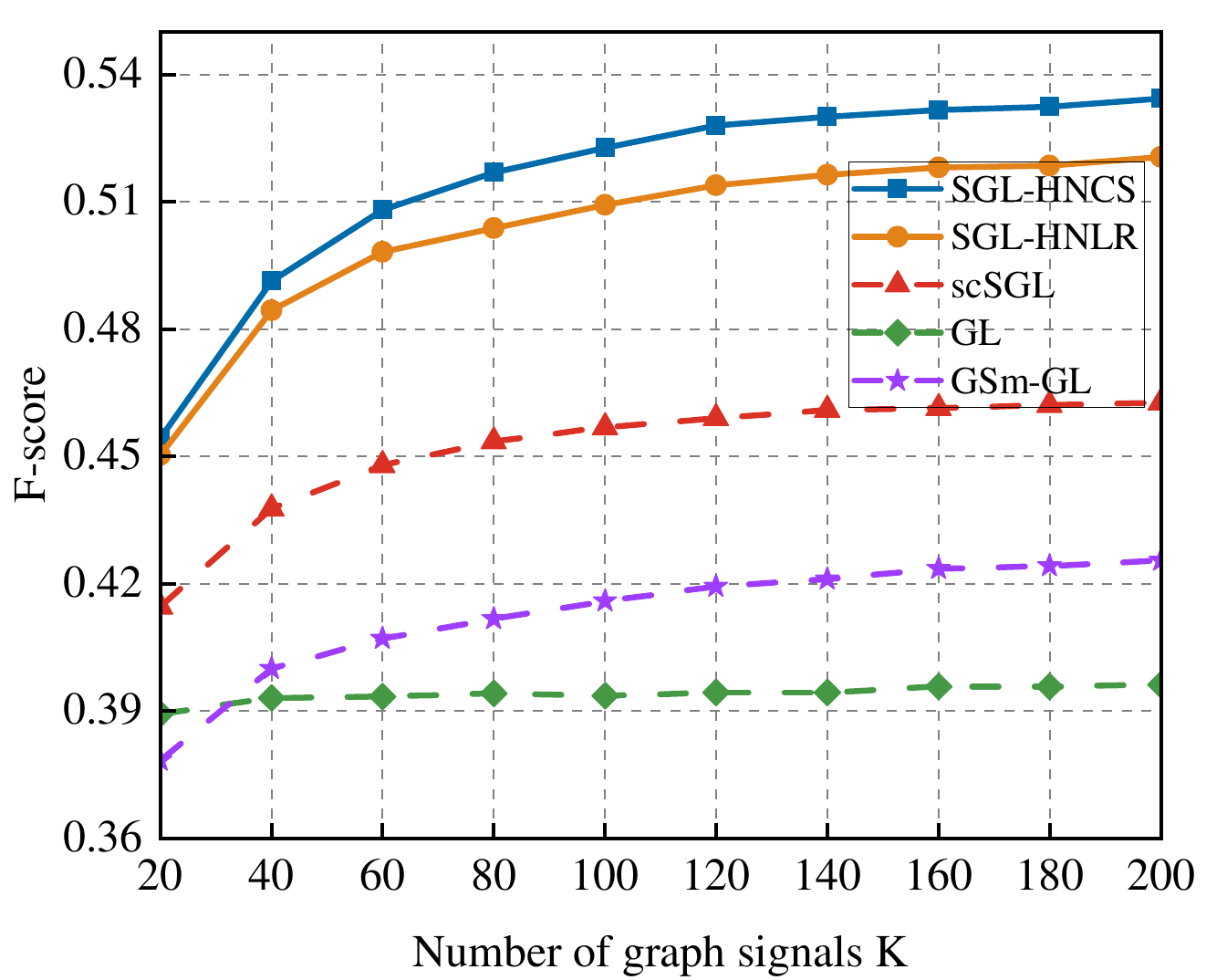}
			\label{f4.1}}
		\hfil
		\subfloat[]{\includegraphics[width=0.48\columnwidth]{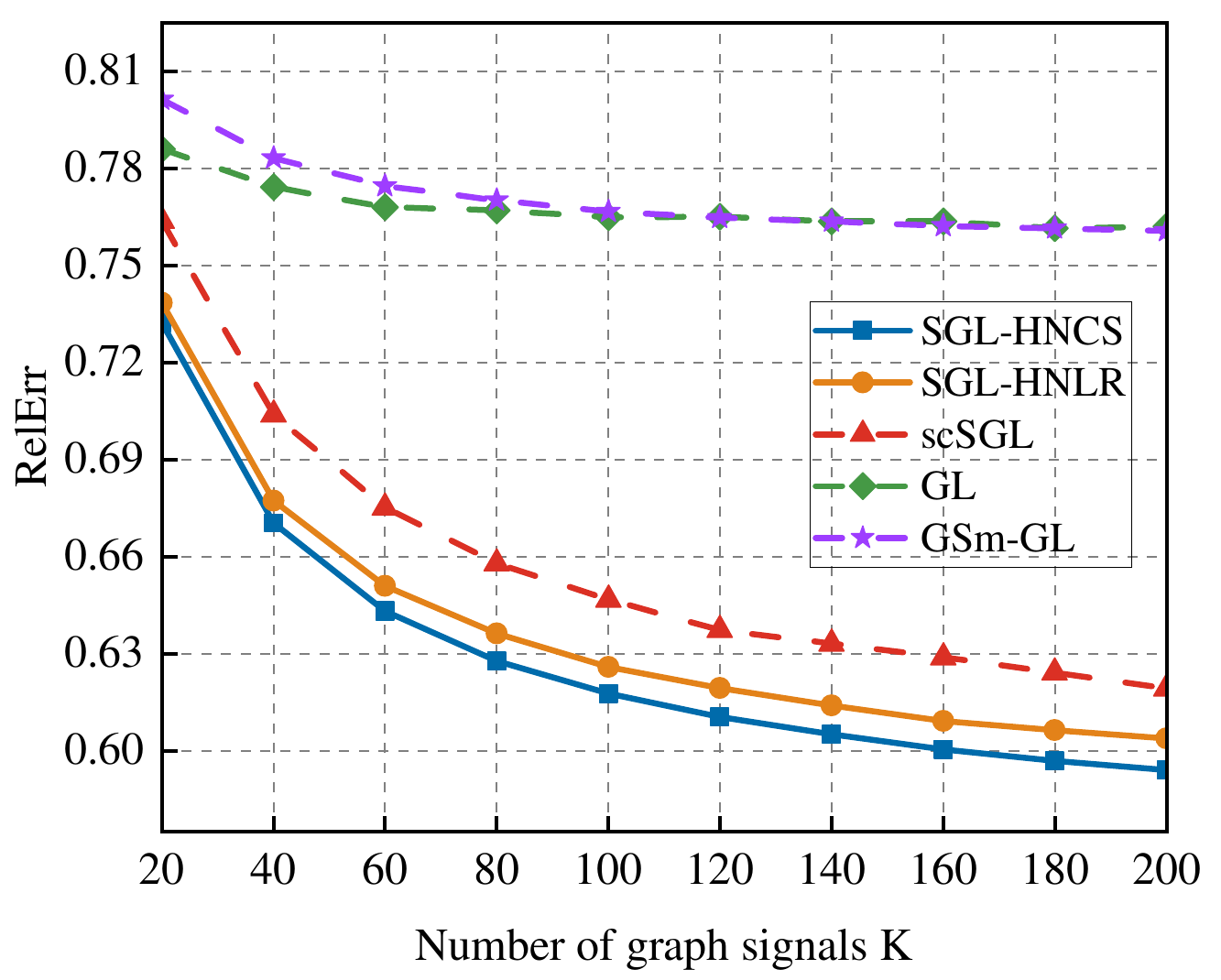}
			\label{f4.2}}
		\caption{Numerical validation of the proposed method under ER graph model with fixed $H=2$. (a) \text{F-score} of the recovered graphs for several methods as the number of graph signals increases. (b) \text{RelErr} of the recovered graphs as the number of graph signals increases for different methods. }
		\label{figr4}
	\end{figure*}
	\begin{figure*}[!t]
		\centering
		\subfloat[]{\includegraphics[width=0.48\columnwidth]{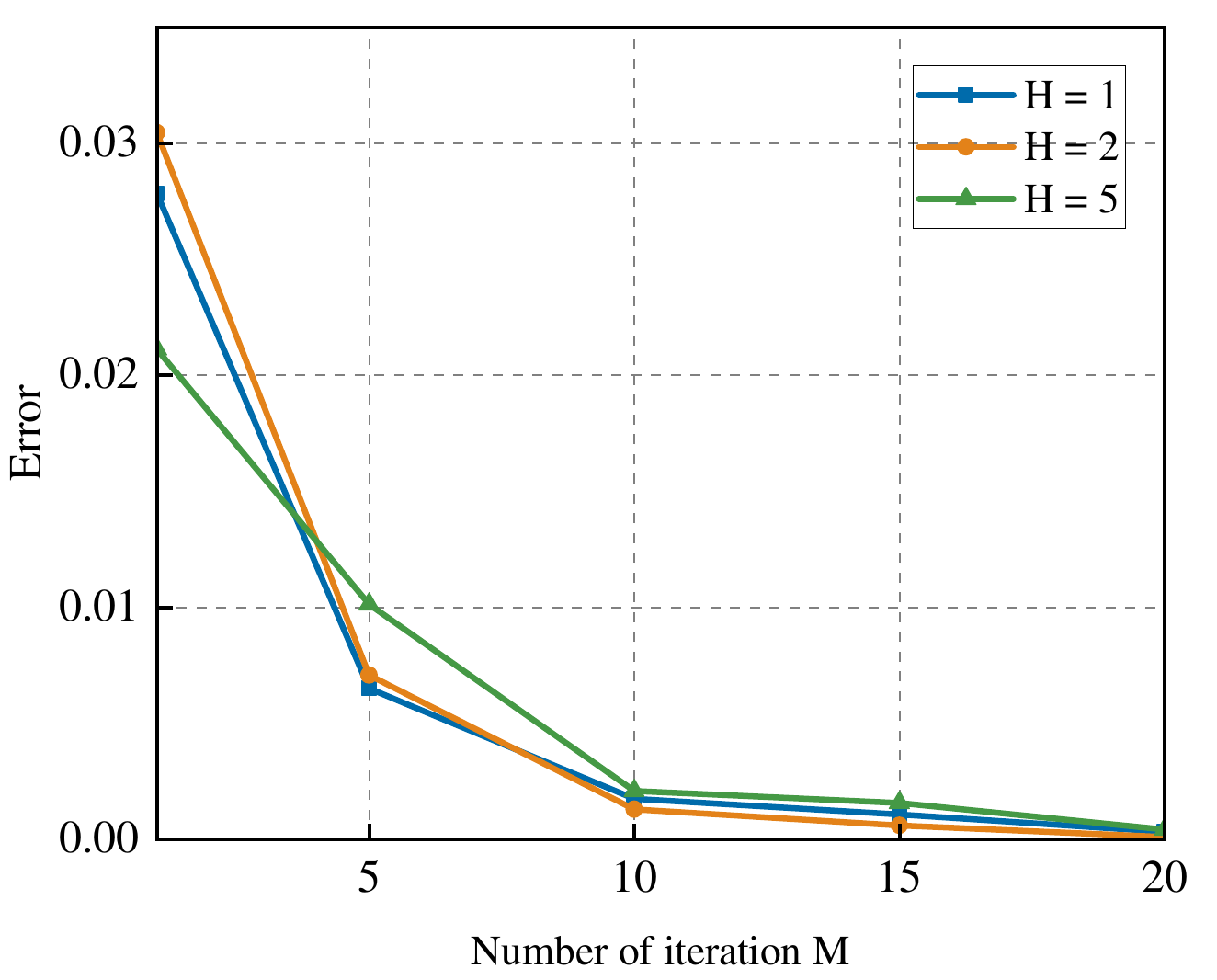}
			\label{conver1}}
		\hfil
		\subfloat[]{\includegraphics[width=0.48\columnwidth]{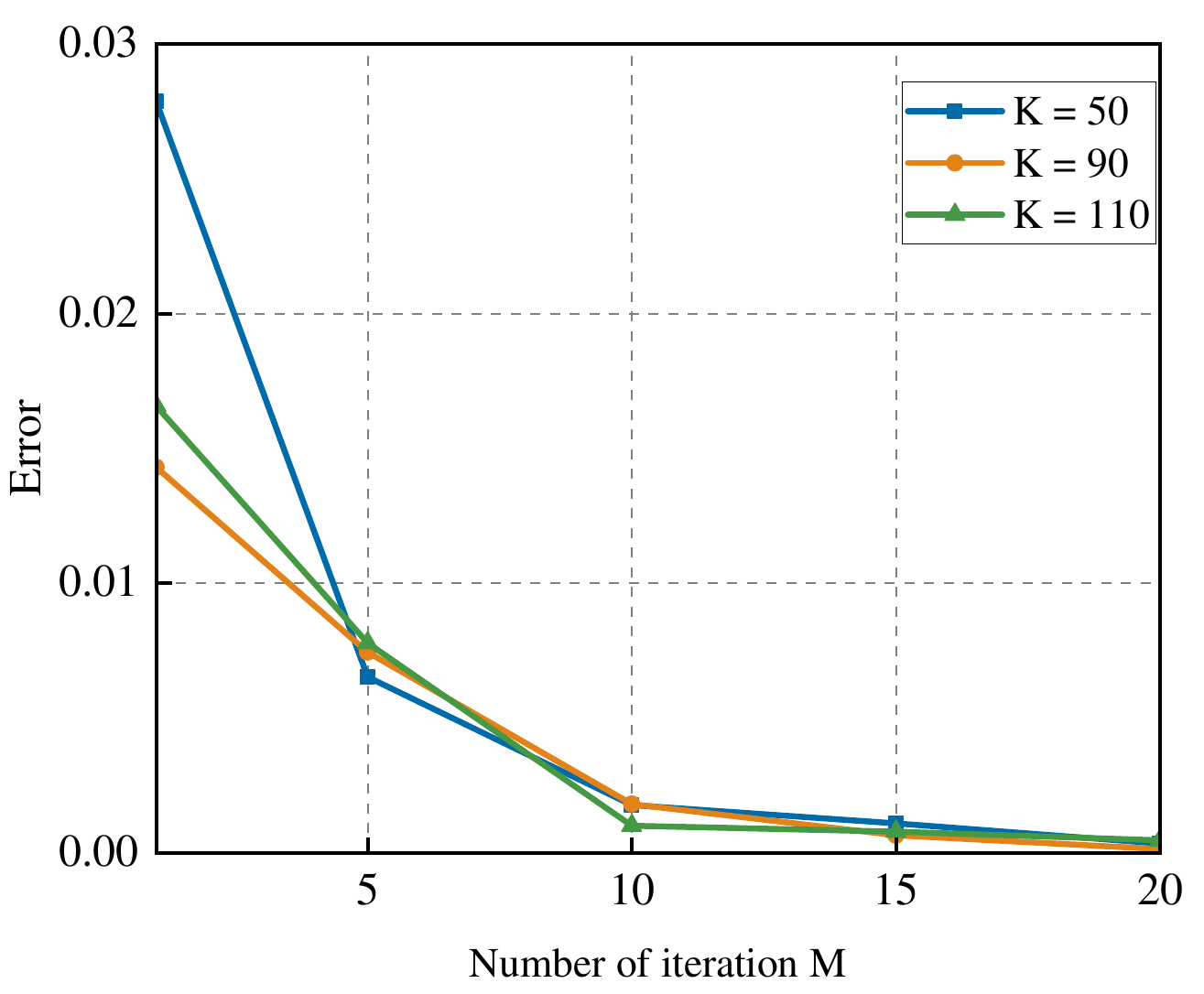}
			\label{conver2}}
		\caption{Results of convergence. (a) the Error metric for different numbers of hidden nodes $H$ as iteration $M$ increases. (b) the Error metric for different numbers of graph signals $K$ as iteration $M$ increases. }
		\label{conver}
	\end{figure*}
In particular, a signed graph with one hidden node is plotted in Fig. \ref{graph}. The visual comparisons are shown in Fig. \ref{figr3}, where the first column corresponds to the groundtruth Laplacian matrices, the second, third and fourth columns present the estimates obtained with $\text{SGL-HNLR}$, $\text{SGL-HNCS}$ and $\text{scSGL}$, respectively. The first row represents the graph Laplacians for $\mathcal{G}^+$, and the second row for $\mathcal{G}^-$. Obviously, the Laplacian matrices learned by $\text{scSGL}$ deviate more from the groundtruth than those learned by our proposed method.
	
\subsubsection{{Number of Graph Signals}}
{The impact of varying the number of graph signals on the performance of different methods is assessed. Specifically, $K$ graph signals are generated with $K$ ranging from 20 to 200, while keeping the number of hidden nodes fixed at $H = 2$.
	
Fig. \ref{figr4} shows the performance comparisons across different values of $K$. The $\text{F-score}$ and $\text{RelErr}$ are plotted in Fig. \ref{f4.1} and Fig. \ref{f4.2}, respectively. As expected, all methods demonstrate improved performance with an increasing number of graph signals. In particular, the existing graph learning methods $\text{GL}$ and \text{GSm-GL} present worse performance than the signed graph learning methods, primarily because they do not consider the positivity and negativity of edge weights. Consistent with the first experiment, $\text{scSGL}$ performs worse than $\text{SGL-HNLR}$ and $\text{SGL-HNCS}$.
\begin{table*}[!t]
		\caption{ Performance of the methods $\text{SGL-HNLR}$, $\text{SGL-HNCS}$, \text{scSGL}, \text{GL}  and $\text{GSm-GL}$ for learning signed graph with hidden nodes on the Wiki and Epinions datasets, respectively. }
		\label{analysis}
		\centering
		\setlength
		\tabcolsep{12pt}
		\vspace{8pt}
		\renewcommand
		\arraystretch{1.4}
		\begin{tabular}{ c| c | c| c| c |c }
			\hline\hline                          
			\textbf{Datasets} &\textbf{Method} & \textbf{F-score} & \textbf{Precision} & \textbf{Recall}& \textbf{NMI}\\
			\hline\hline
			\multirow{6}{*}{}&\text{SGL-HNCS} & 0.5234 & 0.8940	& 0.4905 & 0.4811\\\cline{2-6}
			\multirow{6}{*}{\text{Wiki }}&\text{SGL-HNLR} & 0.5217 & 0.8866 & 0.4849 &	0.4768\\\cline{2-6}
			\multirow{6}{*}{}&\text{scSGL} \cite{kara18} & 0.2284 &	0.5872& 0.5264 &  0.2038\\\cline{2-6}
			\multirow{6}{*}{}&\text{GL \cite{15}} & 0.0902 &	0.7094 & 0.0482 & 0.0633\\\cline{2-6}
			\multirow{6}{*}{}&\text{GSm-GL \cite{Bu27}} & 0.0879 &	0.7109 & 0.0469 & 0.0618\\\cline{2-6}
			\Xhline{1px}
			
			\multirow{6}{*}{Epinions}&\text{SGL-HNCS} & 0.4808 & 0.8745	& 0.4836 & 0.4353\\\cline{2-6}
			\multirow{6}{*}{}
			&\text{SGL-HNLR} & 0.4792 & 0.8945 & 0.4721&	0.4328\\\cline{2-6}
			\multirow{6}{*}{}&\text{scSGL} \cite{kara18} & 0.1989 &	0.3571 & 0.6742 &  0.1618\\\cline{2-6}
			\multirow{6}{*}{}&\text{GL} \cite{15} & 0.1264 &	0.4887 & 0.1402 & 0.0561\\\cline{2-6}
			\multirow{6}{*}{}&\text{GSm-GL} \cite{Bu27} & 0.1111 &	0.5464 & 0.0992 & 0.0537\\\cline{2-6}
			\hline
		\end{tabular}
		\label{tab3}
\end{table*}

\begin{figure*}[!t]
		\centering
		\centering
		\subfloat[]{\includegraphics[width=0.48\columnwidth]{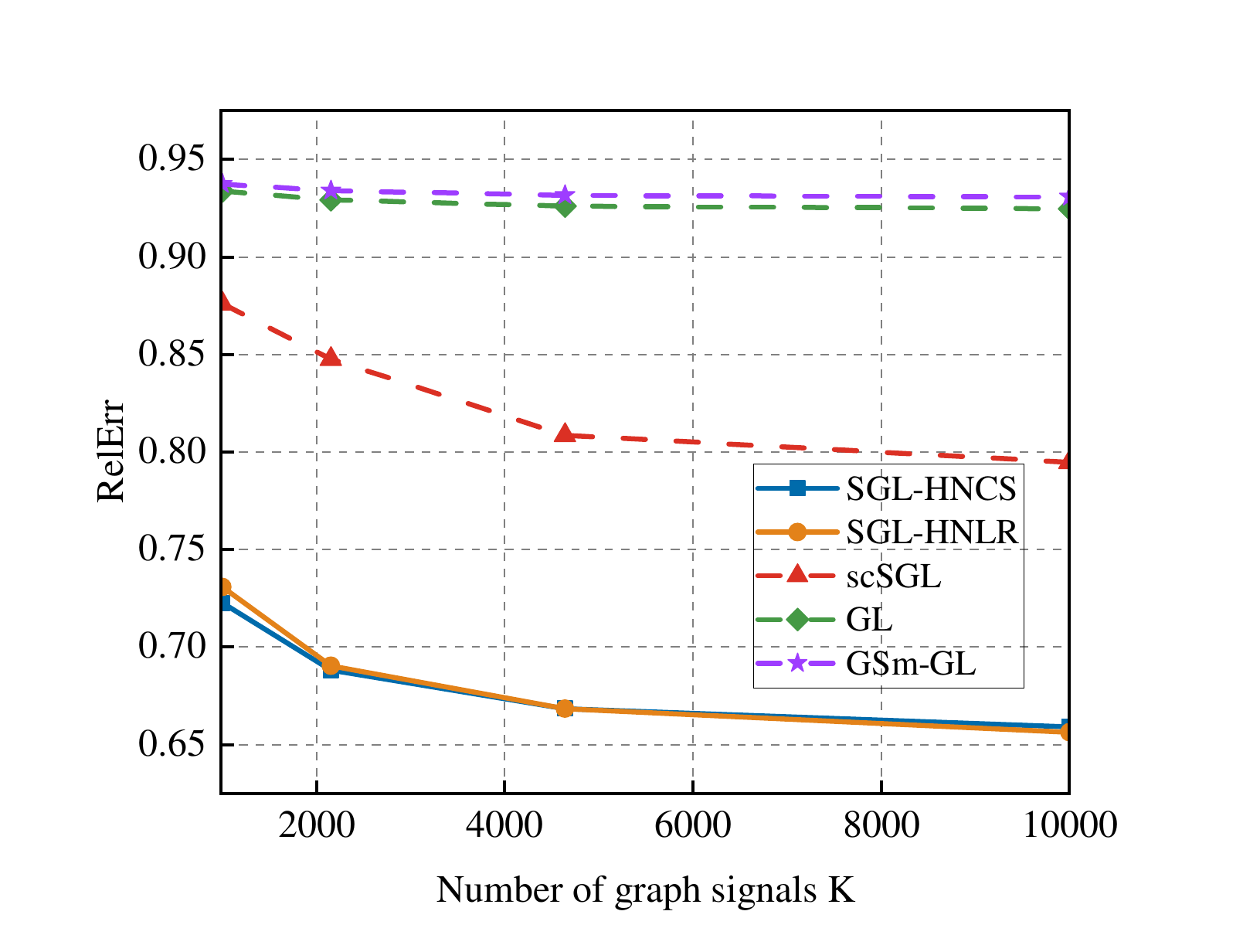}
			\label{f5.1}}
		\hfil
		\subfloat[]{\includegraphics[width=0.48\columnwidth]{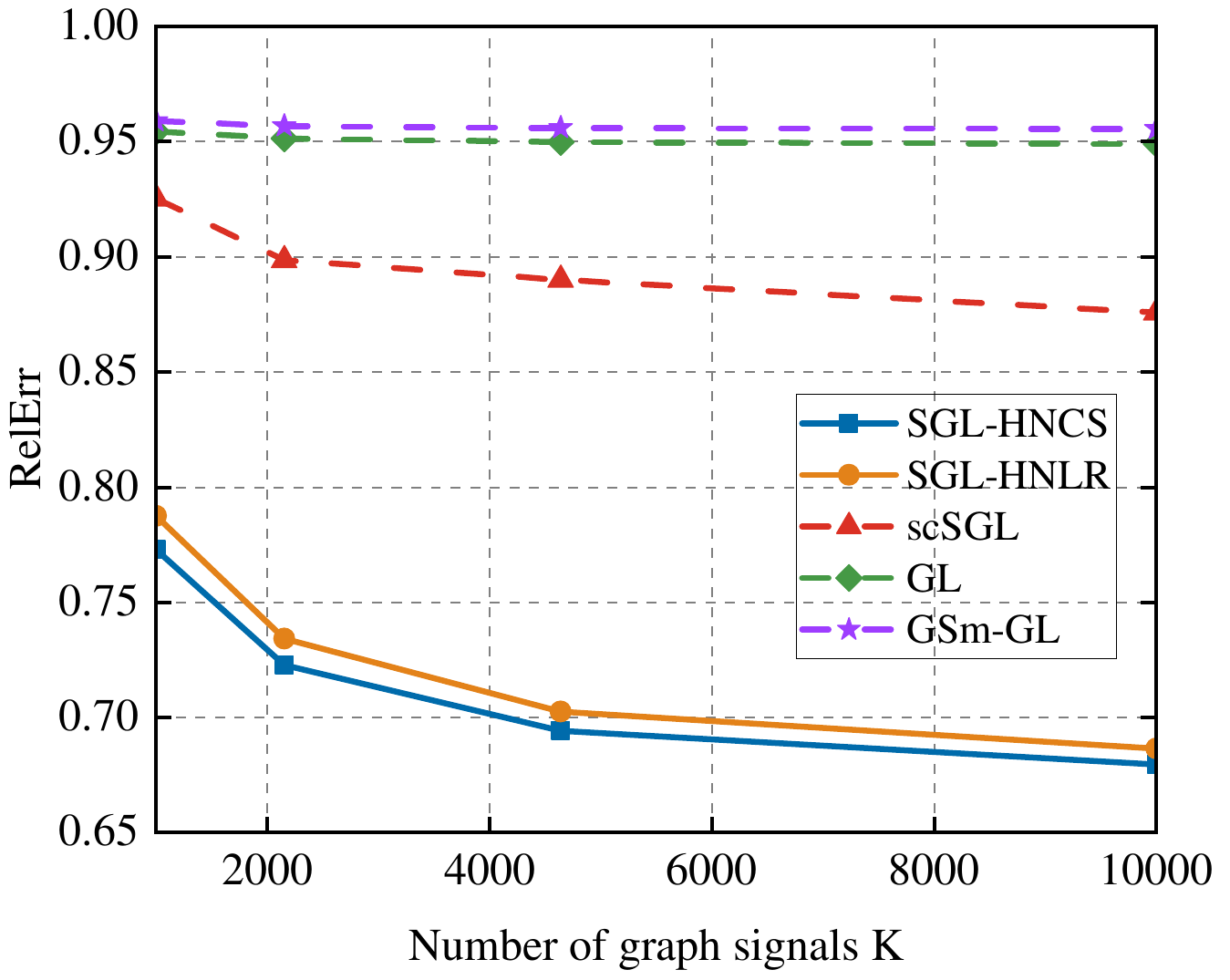}
			\label{f5.2}}
		\caption{RelErr performance of the methods in learning signed graph from two real-world datasets.  The left panel shows the performance for Wiki dataset and the right one shows the performance for Epinions dataset.}
		\label{f5}
\end{figure*} 
	
\subsubsection{Convergence}
Finally, the convergence results of the proposed algorithm with respect to the variables, i.e., the number of hidden nodes $H$ and the number of graph signals $K$, are shown in Fig. \ref{conver1} and Fig. \ref{conver2}. In each figure, the y-axis denotes the relative error of the objective function between two consecutive iterations, i.e., 
$\text{Error} = \left|F_{val}^{(M)}-F_{val}^{(M-1)}\right| \big/ \left|F_{val}^{(M-1)}\right|$, where $F_{val}$ represents the objective function value of (\ref{ex7}). Notably, the results in Fig. \ref{conver} are obtained from a single run of the algorithm. As depicted in Fig. \ref{conver}, the \text{Error} metric decreases sharply within 10 iterations, after which it stabilizes with further iterations. This indicates that the proposed $\text{SGL-HNCS}$ algorithm converges sufficiently.

\subsection{Results on Real-world Data}\label{sec6.3}
This subsection presents the experimental results of our method on two real-world datasets. We use the Wiki dataset \cite{les31} and the Epinions dataset \cite{mas32} to learn the signed graph with hidden nodes. The Wiki dataset is obtained from the Wikipedia site, which describes votes ``for'' (positive) and ``against'' (negative) the other in elections. The Epinions dataset is a ``who trusts whom'' online social network generated from the Epinions site, where one user can either ``trust'' (positive) or ``distrust'' (negative) another. Both datasets contain 401 nodes, but only 396 users are assumed to be observed for inferring the best-represented signed graph from incomplete signals. Thus, the goal is to infer the interactions between these 396 users for the Wiki and Epinions datasets.

The results are shown in Fig. \ref{f5}, where the RelErr decreases with growing the number of samples. This figure demonstrates the efficiency of the proposed method in estimating the graph Laplacian matrix, especially at higher sampling rates. Table \ref{tab3} presents the best $\text{F-score}$, $\text{Precision}$, $\text{Recall}$ and $\text{NMI}$ scores obtained by different methods. Overall, quantitative comparisons indicate that $\text{SGL-HNCS}$ and $\text{SGL-HNLR}$ achieve better performance than $\text{scSGL}$, $\text{GL}$ and $\text{GSm-GL}$. The results once again confirm that the graph learned by our proposed method are more consistent with the groundtruth graph.
\section{Conclusion}\label{sec7}
In this paper, we introduced a novel SGL-HNCS method for learning signed graph in the presence of hidden nodes.  
While previous research has conducted on learning unsigned graphs with hidden nodes, our study represented a pioneering extension to the field of signed graphs. Specifically, we addressed the influence of hidden nodes through a matrix block decomposition approach, predicated on the assumption of signal smoothness in signed graphs. Additionally, we incorporated properties of the underlying graph structure, such as column-sparsity, to formulate the problem of learning signed graph with hidden nodes as an optimization problem. 
Furthermore, we proposed an algorithm to solve this optimization problem using the block coordinate descent (BCD) approach, and the proof of convergence for the proposed algorithm was provided. 
We conducted extensive simulations on synthetic data and real-world data, which demonstrated the effectiveness of the proposed method.
	
Several extensions of this formulation are planned for future work. First, the current formulation assumes that the signed graph is static and does not account for potential smooth changes in the graph structure over time. Second, the assumption regarding hidden nodes in our method is that these nodes are never observed. However, in some applications, the challenge lies in dealing with missing data, where nodes variables may not always be observed. Third, future research could also focus on extending our framework to learn balanced signed graphs with hidden nodes. This involves integrating balance constraints into the framework to improve its applicability in real-world scenarios.

\section*{Acknowledgments}
	This work is supported by the National Natural Science Foundation of China under Grant 62471147.
	\appendix
	\section{Proof of Lemma 1}
	\label{app1}
	\begin{proof}
		To prove Lemma \ref{lem1}, it is necessary to analyze how each term of the objective function in (\ref{ex10}) changes with respect to the input variables in the set $\mathcal{F}$. 
		
		For function $f(\boldsymbol{\ell}^+, \boldsymbol{\ell}^-)$, the two quadratic terms $\alpha_+\langle (2\mathbf{I}+\mathbf{S}^\top\mathbf{S})\boldsymbol{\ell}^+,\boldsymbol{\ell}^+\rangle$ and $\alpha_-\langle (2\mathbf{I}+\mathbf{S}^\top\mathbf{S})\boldsymbol{\ell}^-,\boldsymbol{\ell}^-\rangle$ increase as the norms of $\boldsymbol{\ell}^+$ and $\boldsymbol{\ell}^-$ increase. This is because $(2\mathbf{I}+\mathbf{S}^\top\mathbf{S})$ is positive definite, and two parameters $\alpha_+$, $\alpha_-$ are positive.
		As described in Subsection \ref{up1}, the terms $\langle 2\boldsymbol{c}-\mathbf{S}^\top \boldsymbol{d},\boldsymbol{\ell}^+\rangle$ and $-\langle 2\boldsymbol{c}-\mathbf{S}^\top \boldsymbol{d},\boldsymbol{\ell}^-\rangle$ correspond to the terms $\mathrm{tr}(\hat{\mathbf{C}}_{B}\tilde{\mathbf{L}}_{B}^{+})$ and $-\mathrm{tr}(\hat{\mathbf{C}}_B\tilde{\mathbf{L}}_{B}^{-})$ in (\ref{ex8}). Specifically, as $\hat{\mathbf{C}}_{B}$, $\tilde{\mathbf{L}}_{B}^{+}$ and $\tilde{\mathbf{L}}_{B}^{-}$ are positive semi-definite matrices, it follows that $\langle 2\boldsymbol{c}-\mathbf{S}^\top \boldsymbol{d},\boldsymbol{\ell}^+\rangle\ge 0 $ and $-\langle 2\boldsymbol{c}-\mathbf{S}^\top \boldsymbol{d},\boldsymbol{\ell}^-\rangle\leq 0$. Obviously, the quadratic terms grow faster than the linear terms $\langle 2\boldsymbol{c}-\mathbf{S}^\top \boldsymbol{d},\boldsymbol{\ell}^+\rangle$ and $-\langle 2\boldsymbol{c}-\mathbf{S}^\top \boldsymbol{d},\boldsymbol{\ell}^-\rangle$. As the parameters $\eta{+}$ and $\eta{-}$ are sufficiently large and positive, the growth rate of the logarithmic barrier terms $-\frac{1}{\eta_+}\log[\langle 2\boldsymbol{c}-\mathbf{S}^\top \boldsymbol{d}, \boldsymbol{\ell}^+\rangle+2(\mathbf{1}^\top\boldsymbol{p}^+) +\mathbf{1}^\top\boldsymbol{r}^+]$ and $-\frac{1}{\eta_-}\log[\langle 2\boldsymbol{c}-\mathbf{S}^\top \boldsymbol{d}, \boldsymbol{\ell}^-\rangle+2(\mathbf{1}^\top\boldsymbol{p}^-) +\mathbf{1}^\top\boldsymbol{r}^-]$ becomes slow as variables $(\boldsymbol{\ell}^+,\boldsymbol{\ell}^-)\rightarrow\infty$. Therefore, the logarithmic barrier terms have little impact on the overall behavior of $f(\boldsymbol{\ell}^+, \boldsymbol{\ell}^-)$. In summary, the growth of function $f(\boldsymbol{\ell}^+, \boldsymbol{\ell}^-)$ is dominated by the quadratic terms and we have $f(\boldsymbol{\ell}^+, \boldsymbol{\ell}^-)\rightarrow\infty$ as $(\boldsymbol{\ell}^+,\boldsymbol{\ell}^-)\rightarrow\infty$. The indicator functions $\imath_{\mathcal{S}}(\boldsymbol{\ell}^+)$ and $\imath_{\mathcal{S}}(\boldsymbol{\ell}^-)$ restrict 
		$\boldsymbol{\ell}^+$ and $\boldsymbol{\ell}^-$ to lie within the set 
		$\mathcal{S}$, but they do not influence the coercivity analysis. Therefore, the objective function 
		$h(\boldsymbol{\ell}^+,\boldsymbol{\ell}^-)$ is coercive, ensuring that
		$h(\boldsymbol{\ell}^+,\boldsymbol{\ell}^-)\rightarrow\infty$ as $(\boldsymbol{\ell}^+,\boldsymbol{\ell}^-)\rightarrow\infty$.
		
		Since $\imath_{\mathcal{D}}(\mathbf{v}^+,\mathbf{v}^-)$ only takes finite values when the constraints are satisfied, it does not affect the coercivity within the feasible set $\mathcal{F}$ and we only need to consider the coercivity of $h(\boldsymbol{\ell}^+,\boldsymbol{\ell}^-)$. Consequently, we conclude that $h(\boldsymbol{\ell}^+,\boldsymbol{\ell}^-)+\imath_{\mathcal{D}}(\mathbf{v}^+,\mathbf{v}^-)\rightarrow\infty$ when $\Vert (\mathbf{v}^+,\mathbf{v}^-,\boldsymbol{\ell}^+,\boldsymbol{\ell}^-) \Vert\rightarrow\infty$ and $(\mathbf{v}^+,\mathbf{v}^-,\boldsymbol{\ell}^+,\boldsymbol{\ell}^-)\in\mathcal{F}$. Therefore, Lemma \ref{lem1} holds.
	\end{proof}

	\section{Proof of Lemma 2}\label{app2}
	\begin{proof}
		For convenience, let $\boldsymbol{\ell}=\dbinom{\boldsymbol{\ell}^+}{\boldsymbol{\ell}^-}$ and $\mathbf{v}=\dbinom{\mathbf{v}^+}{\mathbf{v}^-}$, the optimization problem (\ref{ex10}) becomes 
		\begin{align}\label{admm}
			\min\limits_{\mathbf{v},\boldsymbol{\ell}}\quad
			& h(\boldsymbol{\ell})+\imath_{\mathcal{D}}(\mathbf{v})	\notag\\
			\mathrm{s.t.}\quad
			&\mathbf{Av}-\mathbf{B}\boldsymbol{\ell}=\mathbf{0},
		\end{align}		
		where $\mathbf{A}=\mathbf{I}$ and $\mathbf{B}=\mathbf{-I}$ are two coefficient matrices of variable $\mathbf{v}$ and $\boldsymbol{\ell}$, respectively. 
		
		For fixed $\boldsymbol{\ell}$, $\mathbf{Av}=\boldsymbol{\ell}$ has a unique solution due to $\mathbf{A}$ is a full rank matrix. Thus, ${\arg \min}_\mathbf{v}\{h(\boldsymbol{\ell})+\imath_{\mathcal{D}}(\mathbf{v}):\mathbf{Av}=\boldsymbol{\ell}\}$ has a unique minimizer. $\mathit{G}:=\text{Im}(\mathbf{B})\rightarrow\mathbb{R}^q$ defined by $\mathit{G}(\boldsymbol{\ell})\triangleq\operatorname {\arg \min}_\mathbf{v}\{h(\boldsymbol{\ell})+\imath_{\mathcal{D}}(\mathbf{v}):\mathbf{Av}=\boldsymbol{\ell}\}$ is a mapping. For $\forall t_1,t_2\in\mathbb{N}$, it holds that $\Vert \mathit{G}(\boldsymbol{\ell}^{(t_1)})-\mathit{G}(\boldsymbol{\ell}^{(t_2)}) \Vert=\Vert\mathbf{v}^{(t_1)}-\mathbf{v}^{(t_2)}\Vert\leq\bar M \Vert\mathbf{A}\mathbf{v}^{(t_1)}-\mathbf{A}\mathbf{v}^{(t_2)}\Vert$, where $\bar M$ is a constant with $\bar M\ge 1$. Based on the above analysis, $G$ is a Lipschitz continuous mapping. Similarly, for fixed $\mathbf{v}$, ${\arg \min}_{\boldsymbol{\ell}}\{h(\boldsymbol{\ell})+\imath_{\mathcal{D}}(\mathbf{v}):-\mathbf{B}\boldsymbol{\ell}=\mathbf{v}\}$ has a unique minimizer.  $\mathit{F}:=\text{Im}(\mathbf{A})\rightarrow\mathbb{R}^q$ defined by $\mathit{F(\mathbf{v})}\triangleq\operatorname {\arg \min}_{\boldsymbol{\ell}}\{h(\boldsymbol{\ell})+\imath_{\mathcal{D}}(\mathbf{v}):-\mathbf{B}\boldsymbol{\ell}=\mathbf{v}\}$ is a Lipschitz continuous mapping with the fact that $\Vert \mathit{F}(\mathbf{v}^{(t_1)})-\mathit{F}(\mathbf{v}^{(t_2)}) \Vert=\Vert\boldsymbol{\ell}^{(t_1)}-\boldsymbol{\ell}^{(t_2)}\Vert\leq\bar M\Vert\mathbf{-B}\boldsymbol{\ell}^{(t_1)}-(-\mathbf{B}\boldsymbol{\ell}^{(t_2)})\Vert$ for $\forall t_1, t_2\in\mathbb{N}$. Therefore, Lemma \ref{sub_Lc} holds.
	\end{proof}
	\section{Proof of Lemma 3}\label{app3}
	\begin{proof}
		 The indicator functions $\imath_{\mathcal{S}}(\boldsymbol{\ell}^+)$ and $\imath_{\mathcal{S}}(\boldsymbol{\ell}^-)$ are defined as $\imath_{\mathcal{S}}(\boldsymbol{\ell}^+)=\left\{\begin{aligned}
			0 & , & \text{if}\ \boldsymbol{\ell}^+\in\mathcal{S} \\	\infty & , & \text{otherwise}
		\end{aligned}\right.$ and $\imath_{\mathcal{S}}(\boldsymbol{\ell}^-)=\left\{\begin{aligned}
			0 & , & \text{if}\ \boldsymbol{\ell}^-\in\mathcal{S} \\
			\infty & , & \text{otherwise}
		\end{aligned}\right.$. To verify that the function $h(\boldsymbol{\ell}^+,\boldsymbol{\ell}^-)$ is Lipschitz differentiable, we focus on the cases where $\imath_{\mathcal{S}}(\boldsymbol{\ell}^+)$ and $\imath_{\mathcal{S}}(\boldsymbol{\ell}^-)$ are equal to zero. This is because the case where the indicator function takes the value of infinity is irrelevant to the proof. The proof consists of two main steps: first, we establish the differentiability of $h(\boldsymbol{\ell}^+,\boldsymbol{\ell}^-)$, and second, we demonstrate the Lipschitz continuity of its gradient. 

		For the function $h(\boldsymbol{\ell}^{+},\boldsymbol{\ell}^{-})$, the partial derivative with respect to  $\boldsymbol{\ell}^{+}$ and $\boldsymbol{\ell}^{-}$ are given by
		\begin{align}
			&\frac{\partial h(\boldsymbol{\ell}^+,\boldsymbol{\ell}^-)}{\partial \boldsymbol{\ell}^+}=(2\boldsymbol{c}-\mathbf{S}^\top \boldsymbol{d})+2\alpha_+(2\mathbf{I}+\mathbf{S}^\top\mathbf{S})\boldsymbol{\ell}^{+}-\frac{1}{\eta_+}\frac{2\boldsymbol{c}-\mathbf{S}^\top \boldsymbol{d}}{\langle 2\boldsymbol{c}-\mathbf{S}^\top \boldsymbol{d}, \boldsymbol{\ell}^+\rangle	+2(\mathbf{1}^\top\boldsymbol{p}^+) +\mathbf{1}^\top\boldsymbol{r}^+},\\
			&\frac{\partial h(\boldsymbol{\ell}^+,\boldsymbol{\ell}^-)}{\partial \boldsymbol{\ell}^-}=-(2\boldsymbol{c}-\mathbf{S}^\top \boldsymbol{d})+2\alpha_-(2\mathbf{I}+\mathbf{S}^\top\mathbf{S})\boldsymbol{\ell}^{-}-\frac{1}{\eta_-}\frac{2\boldsymbol{c}-\mathbf{S}^\top \boldsymbol{d}}{\langle 2\boldsymbol{c}-\mathbf{S}^\top \boldsymbol{d}, \boldsymbol{\ell}^-\rangle	+2(\mathbf{1}^\top\boldsymbol{p}^-) +\mathbf{1}^\top\boldsymbol{r}^-}.
		\end{align}		
Therefore, the function $h(\boldsymbol{\ell}^+,\boldsymbol{\ell}^-)$ is differentiable.

Next, we verify the gradient Lipschitz continuity of each term in the function $h(\boldsymbol{\ell}^+,\boldsymbol{\ell}^-)$. It is evident that both the linear term $\langle 2\boldsymbol{c}-\mathbf{S}^\top \boldsymbol{d},\boldsymbol{\ell}^+\rangle-\langle 2\boldsymbol{c}-\mathbf{S}^\top \boldsymbol{d},\boldsymbol{\ell}^-\rangle$ and the quadratic term $\alpha_+\langle (2\mathbf{I}+\mathbf{S}^\top\mathbf{S})\boldsymbol{\ell}^+,\boldsymbol{\ell}^+\rangle+\alpha_-\langle (2\mathbf{I}+\mathbf{S}^\top\mathbf{S})\boldsymbol{\ell}^-,\boldsymbol{\ell}^-\rangle$ in the function $h(\boldsymbol{\ell}^+,\boldsymbol{\ell}^-)$ are gradient Lipschitz continuous. Let $h_{log}(\boldsymbol{\ell}^+,\boldsymbol{\ell}^-)=-\frac{1}{\eta_+}\log_{}{[g(\boldsymbol{\ell}^+)]}-\frac{1}{\eta_-}\log_{}{[g(\boldsymbol{\ell}^-)]}$, where $g(\boldsymbol{\ell}^s)=\langle 2\boldsymbol{c}-\mathbf{S}^\top \boldsymbol{d}, \boldsymbol{\ell}^s\rangle	+2(\mathbf{1}^\top\boldsymbol{p}^s) +\mathbf{1}^\top\boldsymbol{r}^s,\forall s\in\{+,-\} $. For any $\boldsymbol{(\ell}_1^+,\boldsymbol{\ell}_1^-)$ and $(\boldsymbol{\ell}_2^+,\boldsymbol{\ell}_2^-)$, we have
\begin{align}
	    &\Big\Vert\Big(\frac{\partial[ h_{log}(\boldsymbol{\ell}_1^+,\boldsymbol{\ell}_1^-)]}{\partial\boldsymbol{\ell}_1^+},\frac{\partial[ h_{log}(\boldsymbol{\ell}_1^+,\boldsymbol{\ell}_1^-)]}{\partial\boldsymbol{\ell}_1^-}\Big)-\Big(\frac{\partial[ h_{log}(\boldsymbol{\ell}_2^+,\boldsymbol{\ell}_2^-)]}{\partial\boldsymbol{\ell}_2^+},\frac{\partial[ h_{log}(\boldsymbol{\ell}_2^+,\boldsymbol{\ell}_2^-)]}{\partial\boldsymbol{\ell}_2^-}\Big)\Big\Vert \notag\\ 
	    =&\Big\Vert\Big(\frac{1}{\eta_{+}}\frac{(\mathbf{S}^\top \boldsymbol{d}-2\boldsymbol{c})}{g(\boldsymbol{\ell}_1^+)},\frac{1}{\eta_{-}}\frac{(\mathbf{S}^\top \boldsymbol{d}-2\boldsymbol{c})}{g(\boldsymbol{\ell}_1^-)}\Big)-\Big(\frac{1}{\eta_{+}}\frac{(\mathbf{S}^\top \boldsymbol{d}-2\boldsymbol{c})}{g(\boldsymbol{\ell}_2^+)},\frac{1}{\eta_{-}}\frac{(\mathbf{S}^\top \boldsymbol{d}-2\boldsymbol{c})}{g(\boldsymbol{\ell}_2^-)}\Big)\Big\Vert\notag\\    
		=&\Big\Vert(2\boldsymbol{c}-\mathbf{S}^\top \boldsymbol{d})\Big\Vert\cdot \Big\Vert \Big(\frac{1}{\eta_+}\big(\frac{1}{g(\boldsymbol{\ell}_1^+)}-\frac{1} {g(\boldsymbol{\ell}_2^+)}\big),  \frac{1}{\eta_-}\big(\frac{1}{g(\boldsymbol{\ell}_1^-)}-\frac{1} {g(\boldsymbol{\ell}_2^-)}\big) \Big)\Big\Vert  \notag\\
		\leq & \bar{M} \Big\Vert\Big(\frac{1}{\eta_+}\frac{\langle 2\boldsymbol{c}-\mathbf{S}^\top \boldsymbol{d} ,\boldsymbol{\ell}_1^+-\boldsymbol{\ell}_2^+\rangle}{g(\boldsymbol{\ell}_1^+)g(\boldsymbol{\ell}_2^+)},\frac{1}{\eta_-}\frac{\langle 2\boldsymbol{c}-\mathbf{S}^\top \boldsymbol{d} ,\boldsymbol{\ell}_1^--\boldsymbol{\ell}_2^-\rangle}{g(\boldsymbol{\ell}_1^-)g(\boldsymbol{\ell}_2^-)}\Big)\Big\Vert\notag\\
		\leq & {\bar{M}^2}\big\Vert  \Big(\frac{1}{\eta_{+}(c_{+})^2}(\boldsymbol{\ell}_1^+-\boldsymbol{\ell}_2^+),\frac{1}{\eta_{-}{(c_{-})^2}} (\boldsymbol{\ell}_1^--\boldsymbol{\ell}_2^-)\Big)\big\Vert,	 		
\end{align}		
where $\bar{M}, c_{+}$ and $c_{-}$ are non-negative constants. The first inequality holds because $2\boldsymbol{c}-\mathbf{S}^\top \boldsymbol{d}$ is a finite vector, and we assume that $\Vert (2\boldsymbol{c}-\mathbf{S}^\top \boldsymbol{d})\Vert\leq\bar{M}$. The second inequality follows from the fact that $g(\boldsymbol{\ell}^s),\forall s\in\{+,-\}$, is a positive and bounded function. Therefore, there exists a Lipschitz constant $M'=\text{max}\Big(\frac{\bar{M}^2}{\eta_{+}(c_{+})^2},\frac{\bar{M}^2}{\eta_{-}(c_{-})^2}\Big)$ such that the log term in the function $h(\boldsymbol{\ell}^+,\boldsymbol{\ell}^-)$ is gradient Lipschitz continuous. 
	
	To this end, the function $h(\boldsymbol{\ell}^+,\boldsymbol{\ell}^-)$ is Lipschitz differentiable with constant $L_h=\text{max}\big(2\alpha_{+}\lambda_{max},\\2\alpha_{-}\lambda_{max}\big)+M'$ where $\lambda_{max}$ is the largest eigenvalue of $2\mathbf{I}+\mathbf{S}^\top\mathbf{S}$, and Lemma \ref{Lip} holds.		
\end{proof} 
\section{Proof of Lemma 4}\label{app4}

\begin{proof}
	To prove Lemma \ref{lem5}, we only need to ensure that the non-smooth parts of the function $\phi(\mathbf{z})$ are separable across the different blocks of variables, as suggested in \cite{hong}. Decomposing $\phi(\mathbf{z})$ as $\phi(\mathbf{z})=\phi_1(\mathbf{z})+\phi_2(\mathbf{z})$, with functions $\phi_1(\mathbf{z})$ and $\phi_2(\mathbf{z})$ being defined as 
	\begin{align}
		\phi_1(\mathbf{z})=&\phi_1(\tilde{\mathbf{L}}_{B}^{+},\tilde{\mathbf{L}}_{B}^{-},\tilde{\mathbf{P}}^{+},\tilde{\mathbf{P}}^{-},\mathbf{R}^{+},\mathbf{R}^{-})\notag\\
		=&\mathrm{tr}(\hat{\mathbf{C}}_{B}\tilde{\mathbf{L}}_{B}^{+})+2\mathrm{tr}(\tilde{\mathbf{P}}^{+})+\mathrm{tr}(\mathbf{R}^{+})-\mathrm{tr}(\hat{\mathbf{C}}_B\tilde{\mathbf{L}}_{B}^{-})\notag\\
		&-2\mathrm{tr}(\tilde{\mathbf{P}}^{-})-\mathrm{tr}(\mathbf{R}^{-})+\alpha_{+}\Vert\tilde{\mathbf{L}}_{B}^{+}\Vert_{F}^2+\alpha_{-}\Vert\tilde{\mathbf{L}}_{B}^{-}\Vert_{F}^2,\\
		\phi_2(\mathbf{z})=&\phi_2(\tilde{\mathbf{P}}^{+},\tilde{\mathbf{P}}^{-})=\sigma_{+}\Vert\tilde{\mathbf{P}}^{+}\Vert_{2,1}+\sigma_{-}\Vert\tilde{\mathbf{P}}^{-}\Vert_{2,1},
	\end{align}
	where $\phi_1(\mathbf{z})$ is a smooth function, while $\phi_2(\mathbf{z})$ is a non-smooth function. However, the non-smooth terms $\Vert\tilde{\mathbf{P}}^{+}\Vert_{2,1}$ and $\Vert\tilde{\mathbf{P}}^{-}\Vert_{2,1}$ in $\phi_2(\mathbf{z})$ only
	involve variables of the second block $\mathbf{z}_2$, it follows that $\phi(\mathbf{z})$ is regular at all feasible points in $\mathcal{Z}^*$. Therefore, Lemma \ref{lem5} holds.
\end{proof}

\bibliographystyle{unsrt}
\bibliography{reference}

\begin{thebibliography}{10}

\bibitem{cam1}
W.~Campbell, C.~Dagli, and C.~Weinstein.
\newblock Social network analysis with content and graphs.
\newblock {\em Lincoln Lab. J.}, 20(1):61--81, 2013.

\bibitem{sand2}
E.~D. Kolaczyk.
\newblock {\em Statistical analysis of network data: Methods and models}.
\newblock Springer, New York, 2009.

\bibitem{lie3}
R.~Li{\'e}geois, A.~Santos, V.~Matta, D.~Van De~Ville, and A.~H. Sayed.
\newblock Revisiting correlation-based functional connectivity and its
  relationship with structural connectivity.
\newblock {\em Netw. Neurosci.}, 4(4):1235--1251, 2020.

\bibitem{An4}
A.~Namaki, A.~Shirazi, R.~Raei, and G.~Jafari.
\newblock Network analysis of a financial market based on genuine correlation
  and threshold method.
\newblock {\em Physica A.}, 390(21):3835--3841, 2011.

\bibitem{Shuman5}
D.~I. Shuman, S.~K. Narang, P.~Frossard, A.~Ortega, and P.~Vandergheynst.
\newblock The emerging field of signal processing on graphs: {E}xtending
  high-dimensional data analysis to networks and other irregular domains.
\newblock {\em IEEE Signal Process. Mag.}, 30(3):83--98, 2013.

\bibitem{sand6}
A.~Sandryhaila and J.~M.~F. Moura.
\newblock Discrete signal processing on graphs.
\newblock {\em IEEE Trans. Signal Process.}, 61(7):1644--1656, 2013.

\bibitem{Marques7}
A.~G. Marques, N.~Kiyavash, J.~M.~F. Moura, D.~Van De~Ville, and R.~Willett.
\newblock Graph signal processing: Foundations and emerging directions [from
  the guest editors].
\newblock {\em IEEE Signal Process. Mag.}, 37(6):11--13, 2020.

\bibitem{leus2023}
G.~Leus, A.~G. Marques, J.~M.~F Moura, A.~Ortega, and D.~I. Shuman.
\newblock Graph signal processing: History, development, impact, and outlook.
\newblock {\em IEEE Signal Process. Mag.}, 40(4):49--60, 2023.

\bibitem{yan2022}
Y.~Yan, E.~E. Kuruoglu, and M.~A. Altinkaya.
\newblock Adaptive sign algorithm for graph signal processing.
\newblock {\em Signal Process.}, 200:108662, 2022.

\bibitem{liu2024}
W.~Liu, H.~Feng, F.~Ji, and B.~Hu.
\newblock Online signed sampling of bandlimited graph signals.
\newblock {\em IEEE Trans. Signal Inf. Process. Netw.}, 10:131--146, 2024.

\bibitem{gir2021}
J.~H. Giraldo, S.~Javed, M.~Sultana, S.~K. Jung, and T.~Bouwmans.
\newblock The emerging field of graph signal processing for moving object
  segmentation.
\newblock In {\em Proc. Int. Workshop Frontiers Comput. Vis.(IW-FCV)}, pages
  31--45, 2021.

\bibitem{dinesh}
C.~Dinesh, G.~Cheung, and I.~V. Baji{\'c}.
\newblock Point cloud denoising via feature graph {L}aplacian regularization.
\newblock {\em IEEE Trans. Image Process.}, 29:4143--4158, 2020.

\bibitem{song}
X.~Song, L.~Chai, and J.~Zhang.
\newblock Graph signal processing approach to {QSAR}/{QSPR} model learning of
  compounds.
\newblock {\em IEEE Trans. Pattern Anal. Mach. Intell.}, 44(4):1963--1973,
  2020.

\bibitem{jin}
J.~Jin, J.~Zhang, J.~Tang, S.~Liang, and Z.~Qu.
\newblock Spatio-temporal data mining with information integrity protection:
  Graph signal based air quality prediction.
\newblock In {\em Proc. IEEE Int. Conf. Acoust., Speech Signal Process.
  (ICASSP)}, pages 5190--5194, 2024.

\bibitem{li2025}
J.~Li, T.~Wan, and W.~Qiu.
\newblock Time-varying sea surface temperature reconstruction leveraging low
  rank and joint smoothness constraints.
\newblock {\em J. Electron. Inf. Techn.}, 47(3):1--9, 2025.

\bibitem{Pavez8}
E.~Pavez and A.~Ortega.
\newblock Generalized {L}aplacian precision matrix estimation for graph signal
  processing.
\newblock In {\em Proc. IEEE Int. Conf. Acoust., Speech Signal Process.
  (ICASSP)}, pages 6350--6354, 2016.

\bibitem{gian9}
G.~B. Giannakis, Y.~Shen, and G.~V. Karanikolas.
\newblock Topology identification and learning over graphs: Accounting for
  nonlinearities and dynamics.
\newblock {\em Proc. IEEE.}, 106(5):787--807, 2018.

\bibitem{Dong10}
X.~Dong, D.~Thanou, M.~Rabbat, and P.~Frossard.
\newblock Learning graphs from data: {A} signal representation perspective.
\newblock {\em IEEE Signal Process. Mag.}, 36(3):44--63, 2019.

\bibitem{Mate11}
G.~Mateos, S.~Segarra, A.~G. Marques, and A.~Ribeiro.
\newblock Connecting the dots: {I}dentifying network structure via graph signal
  processing.
\newblock {\em IEEE Signal Process. Mag.}, 36(3):16--43, 2019.

\bibitem{song2022}
Z.~Song, X.~Yang, Z.~Xu, and I.~King.
\newblock Graph-based semi-supervised learning: A comprehensive review.
\newblock {\em IEEE Trans. Neural Netw. Learn. Syst.}, 34(11):8174--8194, 2022.

\bibitem{egi12}
H.~E. Egilmez, E.~Pavez, and A.~Ortega.
\newblock Graph learning from data under {L}aplacian and structural
  constraints.
\newblock {\em IEEE J. Sel. Topics Signal Process.}, 11(6):825--841, 2017.

\bibitem{kum13}
S.~Kumar, J.~Ying, J.~V. de~Miranda~Cardoso, and D.~P. Palomar.
\newblock Structured graph learning via {L}aplacian spectral constraints.
\newblock {\em Adv. Neural. Inf. Process. Syst.}, 32:11647--11658, 2019.

\bibitem{jav}
A.~Javaheri, A.~Amini, F.~Marvasti, and D.~P. Palomar.
\newblock Learning spatio-temporal graphical models from incomplete
  observations.
\newblock {\em IEEE Trans. Signal Process.}, 72:1361--1374, 2024.

\bibitem{Vk14}
V.~Kalofolias.
\newblock How to learn a graph from smooth signals.
\newblock In {\em Proc. Int. Conf. Artif. Intel. Statist. J. Mach. Learn. Res.
  (PMLR)}, pages 920--929, 2016.

\bibitem{15}
X.~Dong, D.~Thanou, P.~Frossard, and P.~Vandergheynst.
\newblock Learning {L}aplacian matrix in smooth graph signal representations.
\newblock {\em IEEE Trans. Signal Process.}, 64(23):6160--6173, 2016.

\bibitem{baghe}
S.~Bagheri, G.~Cheung, T.~Eadie, and A.~Ortega.
\newblock Joint signal interpolation/time-varying graph estimation via
  smoothness and low-rank priors.
\newblock In {\em Proc. IEEE Int. Conf. Acoust., Speech Signal Process.
  (ICASSP)}, pages 9646--9650, 2024.

\bibitem{seg16}
S.~Segarra, A.~G. Marques, G.~Mateos, and A.~Ribeiro.
\newblock Network topology inference from spectral templates.
\newblock {\em IEEE Trans. Signal Inf. Process. Netw.}, 3(3):467--483, 2017.

\bibitem{dit}
T.~Dittrich and G.~Matz.
\newblock Signal processing on signed graphs: Fundamentals and potentials.
\newblock {\em IEEE Signal Process. Mag.}, 37(6):86--98, 2020.

\bibitem{matz2020}
G.~Matz and T.~Dittrich.
\newblock Learning signed graphs from data.
\newblock In {\em Proc. IEEE Int. Conf. Acoust., Speech Signal Process.
  (ICASSP)}, pages 5570--5574, 2020.

\bibitem{gir}
N.~Girdhar and K.~K. Bharadwaj.
\newblock Signed social networks: {A} survey.
\newblock In {\em Int. Conf. Adv. Comput. Data Sci. (ICACDS)}, pages 326--335,
  2017.

\bibitem{kara18}
A.~Karaaslanli, S.~Saha, S.~Aviyente, and T.~Maiti.
\newblock sc{SGL}: {K}ernelized signed graph learning for single-cell gene
  regulatory network inference.
\newblock {\em Bioinformatics}, 38(11):3011--3019, 2022.

\bibitem{Chan19}
V.~Chandrasekaran, P.~A. Parrilo, and A.~S. Willsky.
\newblock Latent variable graphical model selection via convex optimization.
\newblock In {\em Annu. Allerton Conf. Commun., Control, Comput.(Allerton)},
  pages 1610--1613, 2010.

\bibitem{Chang20}
A.~Chang, T.~Yao, and G.~I. Allen.
\newblock Graphical models and dynamic latent factors for modeling functional
  brain connectivity.
\newblock In {\em Proc. IEEE Data Sci. Wrksp. (DSW)}, pages 57--63, 2019.

\bibitem{Yang21}
X.~Yang, M.~Sheng, Y.~Yuan, and T.~Q.~S. Quek.
\newblock Network topology inference from heterogeneous incomplete graph
  signals.
\newblock {\em IEEE Trans. Signal Process.}, 69:314--327, 2021.

\bibitem{Ana22}
A.~Anandkumar, D.~Hsu, A.~Javanmard, and S.~Kakade.
\newblock Learning linear bayesian networks with latent variables.
\newblock In {\em Proc. Int. Conf. Mach. Learn. (ICML)}, pages 249--257, 2013.

\bibitem{Mei23}
J.~Mei and M.~F. Moura.
\newblock Silvar: Single index latent variable models.
\newblock {\em IEEE Trans. Signal Process.}, 66(11):2790--2803, 2018.

\bibitem{Bu26}
A.~Buciulea, S.~Rey, C.~Cabrera, and A.~G. Marques.
\newblock Network reconstruction from graph-stationary signals with hidden
  variables.
\newblock In {\em Asilomar Conf. Signals, Syst., Comput. (ACSSC)}, pages
  56--60, 2019.

\bibitem{Bu27}
A.~Buciulea, S.~Rey, and A.~G. Marques.
\newblock Learning graphs from smooth and graph-stationary signals with hidden
  variables.
\newblock {\em IEEE Trans. Signal Inf. Process. Netw.}, 8:273--287, 2022.

\bibitem{giraldo2022}
J.~H. Giraldo, A.~Mahmood, B.~Garcia-Garcia, D.~Thanou, and T.~Bouwmans.
\newblock Reconstruction of time-varying graph signals via {S}obolev
  smoothness.
\newblock {\em IEEE Trans. Signal Inf. Process. Netw.}, 8:201--214, 2022.

\bibitem{Ag24}
A.~G. Marques, S.~Segarra, G.~Leus, and A.~Ribeiro.
\newblock Stationary graph processes and spectral estimation.
\newblock {\em IEEE Trans. Signal Process.}, 65(22):5911--5926, 2017.

\bibitem{Per25}
N.~Perraudin and P.~Vandergheynst.
\newblock Stationary signal processing on graphs.
\newblock {\em IEEE Trans. Signal Process.}, 65(13):3462--3477, 2017.

\bibitem{Rey28}
S.~Rey, A.~Buciulea, M.~Navarro, S.~Segarra, and A.~G. Marques.
\newblock Joint inference of multiple graphs with hidden variables from
  stationary graph signals.
\newblock In {\em Proc. IEEE Int. Conf. Acoust., Speech Signal Process.
  (ICASSP)}, pages 5817--5821, 2022.

\bibitem{ye29}
R.~Ye, X.~Q. Jiang, H.~Feng, J.~Wang, R.~Qiu, and X.~Hou.
\newblock Time-varying graph learning from smooth and stationary graph signals
  with hidden nodes.
\newblock {\em EURASIP J. Adv. Signal Process.}, 2024(1):33, 2024.

\bibitem{tse}
P.~Tseng.
\newblock Convergence of a block coordinate descent method for
  nondifferentiable minimization.
\newblock {\em J. Optim. Theory Appl.}, 103(9):475--494, 2001.

\bibitem{cvx30}
M.~Grant and S.~Boyd.
\newblock {CVX}: Matlab software for disciplined convex programming, version
  2.1 beta.
\newblock \url{http://cvxr.com/cvx}, 2013.

\bibitem{sch}
H.~Scheel and S.~Scholtes.
\newblock Mathematical programs with complementarity constraints: Stationarity,
  optimality, and sensitivity.
\newblock {\em Math. Oper. Res.}, 25(1):1--22, 2000.

\bibitem{wanl}
Y.~Wang, W.~Yin, and J.~Zeng.
\newblock Global convergence of {ADMM} in nonconvex nonsmooth optimization.
\newblock {\em J. Sci. Comput.}, 78:29--63, 2019.

\bibitem{fu}
X.~Fu, K.~Huang, M.~Hong, N.~D. Sidiropoulos, and A.~M.~C. So.
\newblock Scalable and flexible multiview {{MAX-VAR}} canonical correlation
  analysis.
\newblock {\em IEEE Trans. Signal Process.}, 65(16):4150--4165, 2017.

\bibitem{erd}
P.~Erd6s and A.~R{\'e}nyi.
\newblock On the evolution of random graphs.
\newblock {\em Publ. Math. Inst. Hungar. Acad. Sci.}, 5(1):17--61, 1960.

\bibitem{les31}
J.~Leskovec, D.~Huttenlocher, and J.~Kleinberg.
\newblock Governance in social media: A case study of the {W}ikipedia promotion
  process.
\newblock In {\em Proc. Int. AAAI Conf. Web Soc. Media. (ICWSM)}, pages
  98--105, 2010.

\bibitem{mas32}
P.~Massa and P.~Avesani.
\newblock Controversial users demand local trust metrics: An experimental study
  on epinions. com community.
\newblock In {\em Proc. Int. AAAI Conf. Web Soc. Media. (ICWSM)}, pages
  121--126, 2005.

\bibitem{hong}
M.~Hong, M.~Razaviyayn, Z.~Luo, and J.~Pang.
\newblock A unified algorithmic framework for block-structured optimization
  involving big data: With applications in machine learning and signal
  processing.
\newblock {\em IEEE Signal Process. Mag.}, 33(1):57--77, 2015.

\end{thebibliography}
	
\end{document}